\documentclass[11pt]{article}
\usepackage[paper=a4paper,layout=letterpaper,hmargin={1in,1in},vmargin=1in]{geometry}
\usepackage[usenames,dvipsnames]{xcolor}
\usepackage{times}
\usepackage[italic]{mathastext}

\usepackage[framed,numbered,autolinebreaks,useliterate]{mcode}
\usepackage[utf8]{inputenc}
\usepackage{bm}
\usepackage{stmaryrd}
\usepackage{latexsym}
\usepackage{amsmath}
\usepackage{amsthm}
\usepackage{amssymb}
\usepackage{graphicx}
\usepackage{ifthen}
\usepackage{calc}
\usepackage{cancel}
\usepackage{enumitem}
\usepackage{hyperref}

\usepackage{longtable}

\usepackage{tikz}
\usetikzlibrary{arrows}
\usetikzlibrary{shapes}
\usetikzlibrary{calc}

\usepackage{algorithmic}

\algsetup{linenodelimiter=.}

\usepackage{float}
\newfloat{algorithm}{ht}{alg}
\floatname{algorithm}{Algorithm}

\allowdisplaybreaks

\makeatletter
\renewcommand{\section}{\@startsection
  {section}%
  {1}%
  {0mm}%
  {-1\baselineskip}%
  {0.5\baselineskip}%
  {\normalfont\large\bfseries}%
}
\renewcommand{\subsection}{\@startsection
  {subsection}%
  {2}%
  {0mm}%
  {-1\baselineskip}%
  {0.5\baselineskip}%
  {\normalfont\large\itshape}%
}

\renewcommand{\subsubsection}{\@startsection
  {subsubsection}%
  {3}%
  {0mm}%
  {-1\baselineskip}%
  {0.5\baselineskip}%
  {\normalfont\itshape}%
}

\newsavebox{\tempbox}
\renewcommand{\@makecaption}[2]{
  \vspace{10pt}
  \sbox{\tempbox}{\textbf{#1.} #2}
  \ifthenelse{\lengthtest{\wd\tempbox > \linewidth}}{
    \textbf{#1.} #2\par
  }{
    \begin{center}
      \textbf{#1.} #2
    \end{center}
  }
}

\makeatother

\numberwithin{equation}{section}

\numberwithin{figure}{section}

\newtheoremstyle{mythm}% name
  {}%      Space above, empty = `usual value'
  {}%      Space below
  {\itshape}% Body font
  {}%         Indent amount (empty = no indent, \parindent = para indent)
  {\bfseries}% Thm head font
  {.}%        Punctuation after thm head
  {.5em}%     Space after thm head: " " = normal interword space;
  {\thmname{#1}~\thmnumber{#2}\ifthenelse{\equal{\thmnote{#3}}{}}{}{~(\thmnote{#3})}}% Thm head spec

\newtheoremstyle{mydefn}% name
  {}%      Space above, empty = `usual value'
  {}%      Space below
  {\upshape}% Body font
  {}%         Indent amount (empty = no indent, \parindent = para indent)
  {\bfseries}% Thm head font
  {.}%        Punctuation after thm head
  {.5em}%     Space after thm head: " " = normal interword space;
  {\thmname{#1}~\thmnumber{#2}\ifthenelse{\equal{\thmnote{#3}}{}}{}{~(\thmnote{#3})}}% Thm head spec

\newtheoremstyle{myremark}% name
  {}%      Space above, empty = `usual value'
  {}%      Space below
  {\upshape}% Body font
  {}%         Indent amount (empty = no indent, \parindent = para indent)
  {\itshape}% Thm head font
  {.}%        Punctuation after thm head
  {.5em}%     Space after thm head: " " = normal interword space;
  {\thmname{#1}~\thmnumber{#2}\ifthenelse{\equal{\thmnote{#3}}{}}{}{~(\thmnote{#3})}}% Thm head spec

\theoremstyle{mythm}
\newtheorem{theo}{Theorem}[section]
\newtheorem{lem}[theo]{Lemma}

\newtheorem{cor}[theo]{Corollary}

\theoremstyle{mydefn}

\newtheorem{exa}[theo]{Example}

\theoremstyle{myremark}

\theoremstyle{mythm}

\newcounter{claimcounter}

\newcommand{\case}[1]{\par\medskip\noindent\textit{Case #1: }}

\newenvironment{cs}{
  \begin{description}
    \renewcommand{\case}[1]{\item[\itshape\mdseries Case ##1:]}
  }{
  \end{description}
}

\newlist{caselist}{description}{10}
\setlist[caselist]{font=\itshape\mdseries}

\setenumerate[1]{label=(\arabic*)}
\newlist{eroman}{enumerate}{2}
\setlist[eroman,1]{label=(\roman*)}
\setlist[eroman,2]{label=(\alph*)}
\newlist{ealph}{enumerate}{1}
\setlist[ealph]{label=(\Alph*)}

\newcounter{nlistcounter}

\renewcommand{\mathbf}[1]{\bm{#1}}
\renewcommand{\phi}{\varphi}
\renewcommand{\epsilon}{\varepsilon}

\newcommand{\RR}{\mathbb R}

\renewcommand{\max}{\operatorname{max}}

\newcommand{\CB}{\mathcal B}
\newcommand{\CC}{\mathcal C}

\newcommand{\CP}{\mathcal P}
\newcommand{\CQ}{\mathcal Q}
\newcommand{\CR}{\mathcal R}
\newcommand{\CS}{\mathcal S}

\renewcommand{\tilde}[1]{\widetilde{#1}}

\newcommand{\Aut}{\operatorname{Aut}}

\newcommand{\ORR}{\overline{\RR}}
\newcommand{\tr}{{\textsf{\upshape t}}}
\newcommand{\str}{{\textsf{\upshape s}}}

\begin{document}

\title{Dimension Reduction via Colour Refinement}
\author{
  Martin Grohe\thanks{Martin Grohe and Erkal Selman were supported by
    the German Research Foundation DFG Koselleck grant GR 1492/14-1.} \\
 \normalsize
  RWTH Aachen University\\
  \normalsize
  {grohe@informatik.rwth-aachen.de}
   \and 
  Kristian Kersting\thanks{Kristian Kersting was
    supported by the Fraunhofer ATTRACT fellowship STREAM and by the
    European Commission under contract number
    FP7-248258-First-MM. Martin Mladenov and Kristian Kersting were
    supported by the German Research Foundation DFG, KE 1686/2-1,
    within the SPP 1527, and the German-Israeli Foundation for
    Scientific Research and Development, 1180-218.6/2011.}\\
  \normalsize
    TU Dortmund University\\
    \normalsize
    {kersting@cs.tu-dortmund.de}
  \and 
  Martin Mladenov\\ 
  \normalsize
    TU Dortmund University\\
    \normalsize
    {martin.mladenov@cs.tu-dortmund.de}
    \and
  Erkal Selman\\
  \normalsize
    RWTH Aachen University\\
    \normalsize
    {selman@informatik.rwth-aachen.de}
  }
\date{}
\maketitle

\begin{abstract}
  Colour refinement is a basic algorithmic routine for graph
  isomorphism testing, appearing as a subroutine in almost all
  practical isomorphism solvers. It partitions the vertices of a graph
  into ``colour classes'' in such a way that all vertices in the same
  colour class have the same number of neighbours in every colour
  class. Tinhofer \cite{tin91}, Ramana, Scheinerman, and Ullman
  \cite{ramschull94} and Godsil \cite{god97} established a tight
  correspondence between colour refinement and fractional
  isomorphisms of graphs, which are solutions to the LP relaxation of
  a natural ILP formulation of graph isomorphism.

  We introduce a version of colour refinement for matrices and extend
  existing quasilinear algorithms for computing the colour
  classes. Then we generalise the correspondence between colour
  refinement and fractional automorphisms and develop a theory of
  fractional automorphisms and isomorphisms of matrices.

  We apply our results to reduce the dimensions of systems of linear
  equations and linear programs.
  Specifically, we show that any given LP $L$ can
  efficiently be transformed into a (potentially) smaller LP $L'$ whose number of variables and constraints is the number
  of colour classes of the colour refinement algorithm, applied to a
  matrix associated with the LP. The transformation is such that
  we can easily (by a linear mapping) map both feasible and
  optimal solutions back and forth between the two LPs. We demonstrate
  empirically that colour refinement can indeed greatly reduce the
  cost of solving linear programs. 
\end{abstract} 

%%%%%%%%%%%%%%%%%%%%%%%%%%%%%%%%%%%%%%%%%%%%%%%%%%%%%%%%%%%%
\section{Introduction}

Colour refinement (a.k.a. ``naive vertex classification'' or ``colour
passing'') is a basic algorithmic routine for graph isomorphism
testing. It iteratively partitions, or colours, the vertices of a
graph according to an iterated degree sequence: initially, all
vertices get the same colour, and then in each round of the iteration
two vertices that so far have the same colour get different colours if
for some colour $c$ they have a different number of neighbours of
colour $c$. The iteration stops if in some step the partition remains
unchanged; the resulting partition is known as the \emph{coarsest
  equitable partition} of the graph. By refining the partition
asynchronously using Hopcroft's strategy of ``processing the smaller
half'' (for DFA-minimisation \cite{hop71}), the coarsest equitable
partition of a graph can be computed very efficiently, in time
$O((n+m)\log n)$ \cite{carcro82,paitar87} (also see \cite{berbongro13}
for a matching lower bound). A beautiful result due to Tinhofer \cite{tin91}, Ramana,
Scheinerman, and Ullman \cite{ramschull94} and Godsil \cite{god97}
establishes a tight correspondence between equitable partitions
of a graph and fractional automorphisms, which are solutions to the LP
relaxation of a natural ILP formulation of graph isomorphism.

In this paper, we introduce a version of colour refinement for
matrices (to be outlined soon) and develop a
theory of equitable partitions and fractional automorphisms and
isomorphisms of matrices. A surprising application of the theory is a
method to reduce the dimensions of systems of linear equations and
linear programs. 

When applied in the context of graph
isomorphism testing, the goal of colour refinement is to partition the
vertices of a graph as finely as possible; ideally, one would like to
compute the partition of the vertices into the orbits of the
automorphism group of the graph. In this paper, our goal is to
partition the rows and columns of a matrix as coarsely as possible. We
show that by ``factoring'' a matrix associated with a system of linear equations or a linear program through
an ``equitable partition'' of the variables and constraints, we
obtain a smaller system or LP equivalent to the original one, in the sense
that feasible and optimal solutions can be transferred back and forth
between the two via linear mappings that we can compute efficiently.
Hence we can use colour refinement as a simple and efficient
preprocessing routine for linear programming, transforming a given
linear program into an equivalent one in a lower dimensional space and
with fewer constraints. We demonstrate the effectiveness of this method experimentally.

Due to the ubiquity of linear programming, our method potentially has
a wide range of applications. Of course not all linear programs show
the regularities needed by our method to be effective. Yet some
do. This work grew out of applications in machine learning, or more
specifically, inference problems in graphical models. Actually, many
problems arising in a wide variety of other fields such as semantic
web, network communication, computer vision, and robotics can also be
modelled using graphical models. The models often have inherent
regularities, which are not exploited by classical inference
approaches such as loopy belief propagation. Symmetry-aware
approaches, see
e.g.~\cite{singla08aaai,kersting09uai,ahmadi2013mlj,bui12arxive}, run
(a modified) loopy belief propagation on the quotient model of the
(fractional) automorphisms of the graphical model and have been proven
successful in several applications such as link prediction, social
network analysis, satisfiability and boolean model counting
problems. Some of these approaches have natural LP formulations, and
the method proposed here is a strengthening of the
symmetry-aware approaches applied by the second and third author
(jointly with Ahmadi) in \cite{mladenov12aistats}.

%%%%%%%%%%%%%%%%%%%%%%%%%%%%%%%%%%%%%%%%%%%%%%%%%%%%%%%%%%%%
\subsection*{Colour Refinement on Matrices}
Consider a matrix $A\in\RR^{V\times W}$.\footnote{We find it
  convenient to index the rows and columns of our matrices by elements
  of finite sets $V,W$, respectively, which we assume to be disjoint. $\RR^{V\times W}$
denotes the set of matrices with real entries and row and column
indices from $V$, $W$, respectively. The order of the rows and columns
of a matrix is irrelevant for us. We denote the entries of a matrix
$A\in\RR^{V\times W}$ by $A_{vw}$.}  We iteratively compute
partitions (or colourings) $\CP_i$ and $\CQ_i$ of the rows and columns of $A$, that is,
of the sets $V$ and $W$. We let $\CP_0=\{V\}$ and $\CQ_0=\{W\}$ be
the trivial partitions. To define $\CP_{i+1}$, we put two rows $v,v'$
in the same class if they are in the same class of $\CP_i$ and if for
all classes $Q$ of $\CQ_i$,
\begin{equation}\label{eq:eqt1}
\sum_{w\in Q}A_{vw}=\sum_{w\in Q}A_{v'w}.
\end{equation}
Similarly, to define $\CQ_{i+1}$, we put two columns $w,w'$
in the same class if they are in the same class of $\CQ_i$ and if for
all classes $P$ of $\CP_i$,
\begin{equation}\label{eq:eqt2}
\sum_{v\in P}A_{vw}=\sum_{v\in P}A_{vw'}.
\end{equation}
Clearly, for some $i\le |V|+|W|$ we have
$(\CP_i,\CQ_i)=(\CP_{i+1},\CQ_{i+1})=(\CP_j,\CQ_j)$ for all
$j\ge i$. We let $(\CP_\infty,\CQ_\infty):=(\CP_{i},\CQ_{i})$.
To see that this is a direct generalisation of colour refinement on
graphs, suppose that $A$ is a $0$-$1$-matrix, and view it as the
adjacency matrix of a bipartite graph $B_A$ with vertex set
$V\cup W$ and edge set $\{vw\mid A_{vw}\neq 0\}$. Then
the coarsest equitable partition of $A$ is equal to the
partition of $V(B_A)$ obtained by running colour refinement on $B_A$
starting from the partition $\{V,W\}$. More generally, we may
view every matrix as a weighted bipartite graph, and thus colour
refinement on matrices is just a generalisation of standard colour
refinement from graphs to weighted bipartite graphs. All
of our results also have a version for arbitrary weighted directed graphs,
corresponding to square matrices, but for the ease of presentation we
focus on the bipartite case here.

Adopting Paige and Tarjan's \cite{paitar87} algorithm for colour
refinement on graphs, we obtain an algorithm
that, given a sparse representation of a matrix $A$, computes
$(\CP_\infty,\CQ_\infty)$ in time $O((n+m)\log n)$, where
$n=|V|+|W|$ and $m$ is the total bitlength of all nonzero entries of $A$
(so that the input size is $O(n+m)$).

Slightly abusing terminology, we say that a \emph{partition} of a
matrix $A\in\RR^{V\times W}$ is a pair $(\CP,\CQ)$ of partitions of
$V$, $W$, respectively. Such a pair partitions the matrix into
``combinatorial rectangles''. A partition $(\CP,\CQ)$ of $A$ is
\emph{equitable} if for all $P\in\CP$, $Q\in\CQ$ and all $v,v'\in P$,
$w,w'\in Q$ equations \eqref{eq:eqt1} and \eqref{eq:eqt2} are
satisfied. It is easy to see that the partition
$(\CP_\infty,\CQ_\infty)$ computed by colour refinement is the
\emph{coarsest} equitable partition, in the sense that it is equitable
and all other
equitable partitions refine it. 

The key result that enables us to apply colour refinement to reduce
the dimensions of linear programs is a correspondence between
equitable partitions and \emph{fractional automorphisms} of a
matrix. We view an \emph{automorphism} of a matrix $A\in\RR^{V\times
  W}$ as a pair of
permutations of the rows and columns that leaves the matrix invariant,
or equivalently, a pair $(X,Y)\in\RR^{V\times V}\times\RR^{W\times W}$
of permutation matrices such that 
\begin{equation}\label{eq:fi}
  XA=AY.
\end{equation}
A \emph{fractional
  automorphism} of $A$ is a pair $(X,Y) \in\RR^{V\times
  V}\times\RR^{W\times W}$ of doubly stochastic matrices
satisfying \eqref{eq:fi}. 
We shall prove (Theorem~\ref{theo:aut}) that every equitable partition
of a matrix yields a fractional automorphism and, conversely, every
fractional isomorphism $(X,Y)$ yields an equitable partition. The
classes of this equitable partition are simply the strongly connected
components of the directed graphs underlying the square matrices $X,Y$.
This basic result is the foundation for everything else in this paper.

We proceed to studying \emph{fractional isomorphisms}
between matrices. Our goal is to be able to compare matrices across
different dimensions, for example, we would like to call the
$(1\times1)$-matrix with entry $2$ and the $(2\times2)$-matrix with four
$1$-entries fractionally isomorphic. The notion of fractional
isomorphism we propose may not be the most obvious one, but we show
that it is fairly robust. In particular, we prove a correspondence between fractional
isomorphisms and \emph{balanced equitable joint partitions} of two
matrices. Furthermore, we prove that fractionally
isomorphic matrices are equivalent when it comes to the solvability of
linear programs.

However, fractional isomorphism is still too fine as an equivalence
relation if we want to capture the solvability of linear programs. We propose
an even coarser equivalence relation between matrices that we call
\emph{partition equivalence}. The idea is that two matrices are
equivalent if they have ``isomorphic'' equitable partitions. We prove that two
linear programs with associated matrices that are partition equivalent
are equivalent in the sense that there are two linear mappings that
map the feasible solutions of one LP to the feasible solutions of the
other, and these mappings preserve optimality.

%%%%%%%%%%%%%%%%%%%%%%%%%%%%%%%%%%%%%%%%%%%%%%%%%%%%%%%%%%%%
\subsection*{Application to Linear Programming}

Every matrix $A$ is partition equivalent to a matrix
$[A]$ obtained by ``factoring'' $A$ through its coarsest equitable
partition; we call $[A]$ the \emph{core factor} of $A$. We can repeat
this factoring process and go to matrices $[[A]],[[[A]]]]$, et
cetera, until we finally arrive at the \emph{iterated core factor}
$\llbracket A\rrbracket$. Now suppose that $A$ is associated with an
LP $L$, then $\llbracket A\rrbracket$ is associated with an LP
$\llbracket L\rrbracket$. To solve $L$, we compute $\llbracket
L\rrbracket$, which we can do efficiently using colour refinement. The
colour refinement procedure also yields the matrices that we need to
translate between the solution spaces of $L$ and $\llbracket
L\rrbracket$.  Then we solve $\llbracket L\rrbracket$ and translate
the solution back to a solution of $L$.

The potential of our method has been confirmed by our computational evaluation on a number of
benchmark LPs with symmetries present. Actually, the time spent
in total on solving the LPs --- reducing an LP and solving the reduced
LP --- is often an order of magnitude smaller than solving the
original LP directly. We have compared our method with a method
of symmetry reduction for LPs due to B\"odi, Grundh\"ofer and
Herr~\cite{bodgruher10}; the experiments show that our method is
substantially faster.

\begin{exa}\label{exa:intro}
  We consider a linear program in standard form:
\begin{equation}
  \tag{$L$}
  \begin{array}{rl}
    \text{min }&c^{\tr} x\\
\text{subject to }&Ax=b,\;x\ge 0,
  \end{array}
\end{equation}
where 
\[
A=\left(
\begin{array}{ccccccccccccc}
  3  &-1 &1  &\frac{1}{4}&\frac{1}{4}&\frac{1}{4}&\frac{1}{4}&0&0&3  &-2 &\frac{1}{2}&\frac{1}{2}\\
  -1 &1  &3  &\frac{1}{4}&\frac{1}{4}&\frac{1}{4}&\frac{1}{4}&0&0&-2 &3  &\frac{1}{2}&\frac{1}{2}\\
  1  &3  &-1 &\frac{1}{4}&\frac{1}{4}&\frac{1}{4}&\frac{1}{4}&0&0&\frac{1}{2}&\frac{1}{2}&\frac{1}{2}&\frac{1}{2}\\
  0  &\frac{1}{3}&\frac{2}{3}&0  &\frac{3}{2}&0  &\frac{3}{2}&2&0&1  &0  &-1 &0\\
  \frac{1}{3}&\frac{1}{3}&\frac{1}{3}&\frac{3}{2}&0  &\frac{3}{2}&0  &2&0&0  &1  &0  &-1\\
  \frac{1}{3}&\frac{1}{3}&\frac{1}{3}&0  &\frac{3}{2}&0  &\frac{3}{2}&0&2&-1 &0  &1  &0\\
  \frac{2}{3}&\frac{1}{3}&0  &\frac{3}{2}&0  &\frac{3}{2}&0  &0&2&0  &-1 &0  &1
\end{array}
\right),\quad
b=
\begin{pmatrix}
  1\\
  1\\
  1\\
  1\\
  1\\
  1\\
  1
\end{pmatrix},
\quad
c=
\begin{pmatrix}
  2\\
  2\\
  2\\
  \frac{3}{2}\\
  \frac{3}{2}\\
  \frac{3}{2}\\
  \frac{3}{2}\\
  1\\
  1\\
  \frac{1}{2}\\
  \frac{1}{2}\\
  \frac{1}{2}\\
  \frac{1}{2}
\end{pmatrix}
\]
We combine $A,b,c$ in a matrix
\[
\tilde A=\left(
\begin{array}{ccc|cccc|cc|cccc|c}
  3  &-1 &1  &\frac{1}{4}&\frac{1}{4}&\frac{1}{4}&\frac{1}{4}&0  &0  &3  &-2 &\frac{1}{2}&\frac{1}{2}&1\\
  -1 &1  &3  &\frac{1}{4}&\frac{1}{4}&\frac{1}{4}&\frac{1}{4}&0  &0  &-2 &3  &\frac{1}{2}&\frac{1}{2}&1\\
  1  &3  &-1 &\frac{1}{4}&\frac{1}{4}&\frac{1}{4}&\frac{1}{4}&0  &0  &\frac{1}{2}&\frac{1}{2}&\frac{1}{2}&\frac{1}{2}&1\\
  \hline
  0  &\frac{1}{3}&\frac{2}{3}&0  &\frac{3}{2}&0  &\frac{3}{2}&2  &0  &1  &0  &-1 &0  &1\\
  \frac{1}{3}&\frac{1}{3}&\frac{1}{3}&\frac{3}{2}&0  &\frac{3}{2}&0  &2  &0  &0  &1  &0  &-1 &1\\
  \frac{1}{3}&\frac{1}{3}&\frac{1}{3}&0  &\frac{3}{2}&0  &\frac{3}{2}&0  &2  &-1 &0  &1  &0  &1\\
  \frac{2}{3}&\frac{1}{3}&0  &\frac{3}{2}&0  &\frac{3}{2}&0  &0  &2  &0  &-1 &0  &1  &1\\
  \hline
  2  &  2&  2&\frac{3}{2}&\frac{3}{2}&\frac{3}{2}&\frac{3}{2}&1&1&\frac{1}{2}&\frac{1}{2}&\frac{1}{2}&\frac{1}{2}&\infty
\end{array}
\right)
\]
by putting $b,c^\tr$ in the last column, row, respectively. 
The lines subdividing the matrix indicate the coarsest equitable
partition. As the core factor of $\tilde A$ we obtain the matrix
\[
[\tilde A]=\left(
\begin{array}{cc|cc|c}
  3&1&0&2&1\\
  1&3&2&0&1\\
  \hline
  6&6&2&2&\infty
\end{array}
\right).
\]
Again, the lines subdividing the matrix indicate the coarsest equitable
partition. The core factor of $[\tilde A]$, which turns out to be the
iterated core factor of $\tilde A$, is
\[
\llbracket \tilde A\rrbracket=[[\tilde A]]=
\begin{pmatrix}
  4&2&1\\
  12&4&\infty
\end{pmatrix}
\]
This matrix corresponds to the LP
\begin{equation}
  \tag{$L'$}
  \begin{array}{rl}
    \text{min }&(c')^{\tr} x'\\
\text{subject to }&A'x'=b',\;x'\ge 0,
  \end{array}
\end{equation}
where 
\[
A'=(4\;\;2),\qquad
b'=(1),\qquad
c'=
\begin{pmatrix}
  12\\
  4
\end{pmatrix}.
\]
An optimal solution to ($L'$) is $x'=(0,\frac{1}{2})^\tr$. To map $x'$ to a
solution of the original LP ($L$), we multiply it with the following
matrix. 
\begin{equation}\label{eq:exa-intro}
D:=
\begin{pmatrix}
  1&0&0&0\\
  1&0&0&0\\
  1&0&0&0\\
  0&1&0&0\\
  0&1&0&0\\
  0&1&0&0\\
  0&1&0&0\\
  0&0&1&0\\
  0&0&1&0\\
  0&0&0&1\\
  0&0&0&1\\
  0&0&0&1\\
  0&0&0&1
\end{pmatrix}
\begin{pmatrix}
  1&0\\
  1&0\\
  0&1\\
  0&1
\end{pmatrix}
=
\begin{pmatrix}
  1&0\\
  1&0\\
  1&0\\
  1&0\\
  1&0\\
  1&0\\
  1&0\\
  0&1\\
  0&1\\
  0&1\\
  0&1\\
  0&1\\
  0&1
\end{pmatrix}
\end{equation}
We will see later where this matrix comes from.
It can be checked that
\[
x:=Dx'=(0,0,0,0,0,0,0,\frac{1}{2},\frac{1}{2},\frac{1}{2},\frac{1}{2},\frac{1}{2},\frac{1}{2})^\tr
\] 
is indeed $x$ is a minimal solution to ($L$).
\end{exa}

%%%%%%%%%%%%%%%%%%%%%%%%%%%%%%%%%%%%%%%%%%%%%%%%%%%%%%%%%%%%
\subsection*{Related Work}
Using automorphisms to speed-up solving optimisation problems has attracted a lot of attention in the literature (e.g. \cite{boedi13,Boyd03fastestmixing,bremner2009polyhedral,Gatermann200495,liberti12,beek05}). 
Most relevant for us is work focusing on integer and
linear programming. For ILPs, methods typically focus on pruning 
the search space to eliminate symmetric solutions, see e.g. \cite{margot10} for a survey). In
linear programming, however, one takes advantage of convexity and projects the LP
into the fixed space of its symmetry group~\cite{boedi13}.
As we will see (in Section~\ref{sec:symm}), our approach subsumes this method.
The second and third author (together with Ahmadi) observed that equitable partitions can compress 
LPs, as they preserve message-passing computations within the log-barrier method~\cite{mladenov12aistats}. 
The present paper builds upon that observation, giving a rigorous theory of dimension reduction using colour-refinement, and connecting to existing symmetry approaches through the notion of fractional automorphisms.
Moreover, we show that the resulting
  theory yields a more general notion of fractional automorphism
  that ties in nicely with the linear-algebra framework and
  potentially leads to even better reductions than the purely
  combinatorial approach of~\cite{mladenov12aistats}.

%%%%%%%%%%%%%%%%%%%%%%%%%%%%%%%%%%%%%%%%%%%%%%%%%%%%%%%%%%%%
\section{Preliminaries}

We use a standard notation for graphs and digraphs.
In a graph $G$, we let $N^G(v)$ denote
the set of neighbours of vertex $v$, and in a digraph $D$ we let
$N_+^D(v)$ and $N_-^D(v)$ denote, respectively, the sets of out-neighbours and
in-neighbours of $v$.

We have already introduced some basic matrix notation in the
introduction. 
A \emph{permutation matrix} is a $0$-$1$-matrix that has exactly one $1$
in every row and column.
We call two matrices $A^1\in\RR^{V^1\times W^1}$ and
$A^2\in\RR^{V^2\times W^2}$ \emph{isomorphic} (and write $A^1\cong
A^2$) if there are bijective mappings $\pi:V^1\to
V^2$ and $\rho:W^1\to W^2$ such that $A^1_{vw}=A^2_{\pi(v)\rho(w)}$ for all
$v\in V^1,w\in W^1$. Equivalently, $A^1$ and $A^2$ are isomorphic if
there are permutation matrices $X\subseteq\RR^{V^2\times V^1}$ and
$Y\subseteq\RR^{W^2\times W^1}$ such that $XA^1=A^2Y$. 

A matrix $X\in\RR^{V\times W}$ is \emph{stochastic} if
it is nonnegative and $\sum_{w\in W}X_{vw}=1$ for all $v\in V$. It is
\emph{doubly stochastic} if both $X$ and its transpose $X^{\tr}$ are
stochastic. Observe that a doubly stochastic matrix is always
square. 

The \emph{direct sum} of two matrices $A^1\in\RR^{V^1\times W^1}$ and
$A^2\in\RR^{V^2\times W^2}$ is the matrix
\[
A^1\oplus A^2:=
\begin{pmatrix}
  A^1&0\\
  0&A^2
\end{pmatrix}
\]
With
every matrix $A\subseteq \RR^{V\times W}$ we associate its
\emph{bipartite graph} $B_A$ with vertex set $V\cup W$ and
edge set $\{vw\mid A_{vw}\neq\emptyset\}$. The matrix $A$ is
\emph{connected} if $B_A$ is connected. (Sometimes, this is called
\emph{decomposable}.) Note that $A$ is not connected if and only if it is
isomorphic to matrix that can be written as the direct sum of two
matrices. A \emph{connected component} of $A$ is a submatrix $A'$ whose rows and
columns form the vertex set of a connected component of the bipartite
graph $B_A$.
With every square matrix
$A\in\RR^{V\times V}$ we associate two more graphs: the directed graph
$D_A$ has vertex set $V$ and edge set $\{(v,v')\mid
A_{vv'}\neq 0\}$. The graph $G_A$ is the underlying undirected graph
of $D_A$. We call $A$ \emph{strongly connected} if the graph $D_A$ is
strongly connected. (Sometimes, this is called \emph{irreducible}.)
It is not hard to see that a doubly stochastic matrix $X$ is strongly
connected if and only if the graph $G_A$ is connected.

Let $A\in\RR^{V\times
  W}$. For all subsets
$P\subseteq V, Q\subseteq W$, we let
\begin{equation}\label{eq:e0}
F^A(P,Q)=\sum_{(v,w)\in P\times Q}A_{vw}.
\end{equation}
If we interpret $A$ as a weighted bipartite graph, then $F^A(P,Q)$ is the total weight of the edges
from $P$ to $Q$.
We write
$F^A(v,Q)$, $F^A(P,w)$ instead of $F^A(\{v\},Q)$, $F^A(P,\{w\})$.
Recall that a \emph{partition of $A$} is a
pair $(\CP,\CQ)$, where $\CP$ is a partition of the set $V$ of row
indices and $\CQ$ is a partition of the set $W$ of column indices. 
Using the function $F$, we can express the conditions \eqref{eq:eqt1}
and \eqref{eq:eqt2} for a partition
being \emph{equitable} as 
\begin{align}
\label{eq:e1}
F(v,Q)&=F(v',Q)&\text{for all }v,v'\in P;\\
\label{eq:e2}
F(P,w)&=F(P,w')&\text{for all }w,w'\in Q.
\end{align}
for all
$P\in\CP,Q\in \CQ$.

A \emph{convex combination} of numbers $a_i$ is a sum $\sum_i\lambda_i
a_i$ where $\lambda_i\ge 0$ for all $i$ and $\sum_i\lambda_i=1$.  If
$\lambda_i>0$ for all $i$, we call the convex combination
\emph{positive}. We need the following simple (and well-known) lemma
about convex combinations.

\begin{lem}\label{lem:invisible_hand}
  Let $D$ be a strongly connected digraph.  Let $f: V(D)\rightarrow
  \RR$, such that for every $v\in V(D)$, the number $f(v)$ is a
  positive convex combination of all $f(w)$ for $w\in N_+(v)$.  Then
  $f$ is constant.
\end{lem}

\begin{proof}
  Suppose for contradiction that $f$ satisfies the
  assumptions, but is not constant.  Let $v\in V(D)$ be a vertex with
  maximum value $f(v)$ and $w\in V(D)$ such that $f(w)<f(v)$. Let $P$
  be a path from $v$ to $w$. Then $P$ contains an edge $v'w'$ such
  that $f(v)=f(v')>f(w')$. By the maximality of $f(v)$, for all
  $w''\in N_+(v')$ it holds that $f(v')\ge f(w'')$, and this
  contradicts $f(v')$ being a positive convex combination of the
  $f(w'')$ for $w''\in N_+(v')$.
\end{proof}

Sometimes, we consider matrices with entries from
$\ORR=\RR\cup\{\infty\}$. We will only form linear combinations of
elements of $\ORR$ with nonnegative real coefficients, using the rules
$r+\infty=\infty+r=\infty$ for all $r\in\ORR$ and $0\cdot\infty=0$,
$r\cdot\infty=\infty$ for $r>0$.

All our results hold for rational and real matrices and vectors. For the
algorithms, we assume the input matrices and vectors to be
rational. To analyse the algorithms, we use a standard RAM model.

%%%%%%%%%%%%%%%%%%%%%%%%%%%%%%%%%%%%%%%%%%%%%%%%%%%%%%%%%%%%
\section{Colour Refinement in Quasilinear Time}\label{sec:alg}

In this section, we describe an algorithm that computes the coarsest
equitable partition of a matrix $A\in\RR^{V\times W}$ in time
$O((n+m)\log n)$. Here $n:=|V|+|W|$ and $m$ is the total bitlength
of all nonzero entries of $A$. (We use this notation for the rest of
this section.)

To
describe the algorithm, we view $A$ as a weighted bipartite graph with
vertex set $V\cup W$ and edges with nonzero weights representing
the nonzero matrix entries. For every vertex $u\in V\cup W$ and every
set $C\subseteq V\cup W$ of
vertices, we let $F^*(u,C)$ be the sum of the weights of the edges incident
with $u$. That is, $F^*(v,C)=F(v,C\cap W)$ for $v\in V$ and
$F^*(w,C)=F(C\cap V,w)$ for $w\in W$. Moreover, for every subset
$C\subseteq V\cup W$, we let $m_C$ be the total bitlength of the weight
of all edges incident with a vertex in $C$. For a vertex $u$, we write
$m_u$ instead of
$m_{\{u\}}$. Note that $m=\sum_{u\in V}m_u=\sum_{u\in W}m_u$.

We consider the problem of computing the coarsest equitable partition
of $A$. 
A naive implementation of the iterative refinement procedure described
in the introduction would yield a running time that is (at least)
quadratic: in the worst case, we need $n$ refinement
rounds, and each round takes time  $\Omega(n+m)$.

A significant improvement can be achieved if the refinement steps are
carried out asynchronously, using a strategy that goes back to
Hopcroft's algorithm for minimising deterministic finite automata
\cite{hop71}. The idea is as follows. The algorithm
maintains partitions $\CC$ of $V\cup W$. We call the classes of $\CC$
\emph{colours}. Initially,
$\CC=\{V,W\}$. Furthermore, the algorithm keeps a stack $S$ that holds
some colours that we still want to use for refinement in the
future. Initially, $S$ holds $V,W$ (in either order). In each
refinement step, the algorithm pops a colour $D$ from the stack. We
call $D$ the \emph{refining colour} of this refinement step. For
all $u\in V\cup W$ we compute the value $F^*(v,D)$. Then
for each colour $C$ in the current partition that has at least one
neighbour in $D$, we partition $C$ into new classes $C_1,\ldots,C_k$
according to the values $F^*(u,D)$. Then we replace $C$ by
$C_1,\ldots,C_k$ in the partition $\CC$. Moreover, we add all classes
among $C_1,\ldots,C_k$ \emph{except for the largest} to the stack
$S$. If we use the right data structures, we can carry out such a
refinement step with refining colour $D$ in time $O(|D|+m_D)$. Compared to the standard,
unweighted version of colour refinement, the weights add some
complication when it comes to computing the partition $C_1,\ldots,C_k$
of $C$. We can handle this by standard vector partitioning techniques,
running in time linear in the total bitlength of the weights involved.
By not adding the largest
among the classes $C_1,\ldots,C_k$ to the stack, we achieve that every
vertex $u$ appears at most $\log n$ times in a refining colour
$D$. Whenever $u$ appears in the refining colour, it contributes $O(1+m_u)$
to the cost of that refinement step. Thus the overall cost is
$\sum_{u\in V\cup W}O(1+m_u)\log n= O((n+m)\log n)$. We refer the
reader to \cite{berbongro13,paitar87} for details on the algorithm
(for the unweighted case) and
its analysis.

\begin{theo}\label{theo:alg}
  There is an algorithm
that, given a sparse representation of a matrix $A$, computes the
coarsest equitable partition of $A$ in time $O((n+m)\log n)$.
\end{theo}

%%%%%%%%%%%%%%%%%%%%%%%%%%%%%%%%%%%%%%%%%%%%%%%%%%%%%%%%%%%%
\section{Fractional Automorphisms}

Recall that a \emph{fractional automorphism} of a matrix
$A\in\RR^{V\times W}$ is a pair $(X,Y)\in\RR^{V\times
  V}\times\RR^{W\times W}$ of doubly stochastic matrices such that
\begin{equation}\label{eq:fi0}
XA=AY.
\end{equation}
In this section, we prove the theorem relating fractional
automorphisms to equitable partitions. For every pair $(X,Y)\in\RR^{V\times
  V}\times\RR^{W\times W}$ of matrices we let $\CP_X$ be the partition
of $V$ into the strongly connected components of $X$, and we let $\CQ_Y$ be the partition
of $W$ into the strongly connected components of $Y$. Conversely, for every
partition $(\CP,\CQ)$ of $A$, we let $X_{\CP}\in\RR^{V\times V}$\label{page:XP} be the
matrix with entries $X_{vv'}:=1/|P|$ if $v,v'\in P$ for some $P\in\CP$
and $X_{vv'}:=0$ otherwise, and we let $Y_{\CQ}\in\RR^{W\times W}$ be the
matrix with entries $Y_{ww'}:=1/|Q|$ if $w,w'\in Q$ for some $Q\in\CQ$
and $Y_{vv'}:=0$ otherwise.

\begin{theo}\label{theo:aut}
  Let $A\in\RR^{V\times W}$.
  \begin{enumerate}
  \item If $(\CP,\CQ)$ is an equitable partition of $A$, then
    $(X_\CP,Y_\CQ)$ is a fractional automorphism.
  \item If $(X,Y)$ is a fractional automorphism of $A$, then
    $(\CP_X,\CQ_Y)$ is an equitable partition.
  \end{enumerate}
\end{theo}

\begin{proof}[Proof of Theorem~\ref{theo:aut}]
  To prove (1), let $(\CP,\CQ)$ be an equitable partition of $A$, and
  let $X:=X_\CP$ and $Y:=Y_{\CQ}$. Let $v\in V,w\in
    W$, and let $P\in\CP$ and $Q\in\CQ$ be the
    classes of $v$ and $w$, respectively.  Then
    \begin{equation}\label{eq:e5}
    (X A)_{vw} =
    \sum_{v'\in P} \frac{1}{|P|} \cdot A_{v'w} =
    \frac{1}{|P|} \cdot F(P,w) \overset{(a)}{=}
    \frac{1}{|Q|} \cdot F(v,Q) =
    \sum_{w'\in Q} A_{vw'}\cdot\frac{1}{|Q|} =
    (A Y)_{vw}
    \end{equation}
    {\sloppy
    Equality $(a)$ can be established by a double-counting
    argument:  we have
    $
    F(P,Q) =\sum_{v'\in P} F(v',Q) = |P|\cdot F(v,Q)
    $
    by \eqref{eq:e1} and
    $
    F(P,Q) = \sum_{w'\in Q} F(P,w') = |Q|\cdot F(P,w)
    $
    by \eqref{eq:e2}.
  }

    To prove (2), let $(X,Y)$ be a fractional automorphism of $A$.  Let
    $P\in\CP_X$ and $Q\in\CQ_Y$. We need to prove that $P,Q$ satisfy
    \eqref{eq:e1} and \eqref{eq:e2}. 

    We first prove
    \eqref{eq:e1}.
    For every $v\in P$, we have
    \begin{equation}\label{eq:e6}
    \begin{split}
      F(v,Q)
      &= \sum_{w'\in Q} A_{vw'}
      \overset{(b1)}{=} \sum_{w'\in Q} A_{vw'} \sum_{w\in Q} Y_{w'w}
      = \sum_{w\in Q}\sum_{w'\in Q} A_{vw'} Y_{w'w}  \\
      &\overset{(b2)}{=} \sum_{w\in Q}\sum_{v'\in P} X_{vv'} A_{v'w}
      = \sum_{v'\in P} X_{vv'} \underbrace{\sum_{w\in Q} A_{v'w}}_{=F(v',Q)}
      \overset{(b3)}{=} \sum_{v'\in N^{D_X}_+(v)} X_{vv'} \cdot F(v',Q),
    \end{split}
    \end{equation}
    Equation $(b1)$ holds because $\sum_{w\in W}Y_{w'w}=1$ and $Y_{w'w}=0$
    for $w'\in Q,w\not\in Q$. Here we use that $Q$, which by
    definition is a strongly connected component of the digraph $D_Y$,
    is also a connected component of the undirected graph $G_Y$.
    Equation $(b2)$ holds by $XA=AY$ and $Y_{w'w}=0$ for $w'\not\in Q,
    w\in Q$ and $X_{vv'}=0$ for $v\in P,v'\not\in P$. Equation $(b3)$
    holds, because $N^{D_X}_+(v)\subseteq P$ and
    $X_{vv'}\neq 0\iff v'\in N^{D_X}_+(v)$.
    As the matrix $X$ is stochastic, this implies that $F(v,Q)$
    is a positive convex combination of the $F(v',Q)$ for $v'\in
    N^{D_X}_+(v)$. As $P$ is the vertex set of a strongly connected
    component of $D_X$, by Lemma~\ref{lem:invisible_hand}, it follows that
      $F(v,Q)=F(v',Q)$ for all $v,v'\in P$. This proves
      \eqref{eq:e1}.

      \eqref{eq:e2} can be proved similarly.
   \end{proof}

%%%%%%%%%%%%%%%%%%%%%%%%%%%%%%%%%%%%%%%%%%%%%%%%%%%%%%%%%%%%
\section{Fractional Isomorphisms}

In this section, we want to relate different matrices by ``fractional
isomorphisms''.
Let us first review the natural notion of fractional isomorphisms of
graphs: if $G^1$ and $G^2$ are (undirected) graphs with vertex sets
$V^1,V^2$, respectively, and $A^1\in\{0,1\}^{V^1\times V^1}$,
$A^2\in\{0,1\}^{V^2\times V^2}$ are their adjacency matrices, then $G^1$
and $G^2$ are \emph{fractionally isomorphic} if there is a doubly
stochastic matrix $X\in\RR^{V^2\times V^1}$ such that
$XA^1=A^2X$. Note that this implies $|V^1|=|V^2|$, because doubly
stochastic matrices are square. Viewing matrices as weighted bipartite
graphs, it is straightforward to generalise this notion of fractional
isomorphism to pairs of matrices of the same dimensions: let $A^1\in\RR^{V^1\times W^1}$,
$A^2\in\RR^{V^2\times W^2}$ with $|V^1|=|V^2|$ and $|W^1|=|W^2|$. We
could call $A^1$ and $A^2$ fractionally isomorphic if there is a
pair $(X,Y)\in\RR^{V^2\times V^1}\times\RR^{W^2\times W^1}$ of doubly
stochastic matrices such that
\begin{align}
  \label{eq:i1a}
  XA^1&=A^2Y\\
   \label{eq:i1b}
   X^{\tr}A^2&=A^1Y^{\tr}.
\end{align}
(We need both equations because the full adjacency matrix of the
weighted bipartite graph represented by a matrix $A$ is $
\begin{pmatrix}
  0&A\\A^{\tr}&0
\end{pmatrix}
$.) If we were only interested in fractional isomorphisms between
matrices of the same dimensions, this would be a perfectly reasonable
definition. It is not clear, though, how to generalise it to
matrices of different dimensions, and in a sense this whole paper is
about similarities between matrices of different dimensions.  If the
matrices have different dimensions, we need to drop the requirement of
$X,Y$ being doubly stochastic. The first idea would be to require
$X,Y$ to be
stochastic matrices with constant column sums. The resulting notion of
fractional isomorphism is the right one for connected
matrices, but
has the backdraw that it is not closed under direct sums. That is,
there are matrices $A^1,A^2,A^3,A^4$ such that $A^1,A^2$ and $A^3,A^4$
are fractionally isomorphic in this sense, but the direct sums
$A^1\oplus A^3$ and $A^2\oplus A^4$ are not (see
Example~\ref{exa:equi}). However, closure under direct sums  is something we might expect of
something called ``isomorphism''.

There is an alternative approach to defining fractional
isomorphisms that is more robust, but 
equivalent for connected matrices. The starting point is the
observation that isomorphisms between two graphs correspond to automorphisms
of their disjoint union where each vertex of the first graph is mapped
to a vertex of the second and vice versa. Replacing automorphisms by
fractional automorphisms, this leads to the
following definition. As above, we consider matrices $A^1\in\RR^{V^1\times W^1}$ and $A^2\in\RR^{V^2\times
  W^2}$, where we
assume the sets $V^1,V^2,W^1,W^2$ to be mutually
disjoint. A \emph{fractional isomorphism} from $A^1$ to $A^2$ is a
fractional automorphism $(X,Y)$ of the direct sum
\[
A^1\oplus A^2=\begin{pmatrix}A^1&0\\0&A^2\end{pmatrix}\in\RR^{(V^1\cup
  V^2)\times(W^1\cup W^2)}
\]
such that, for $j=1,2$, for every $v^j\in V^j$ there is a $v^{3-j}\in
V^{3-j}$ with $X_{v^jv^{3-j}}\neq 0$ and for every $w^j\in W^j$ there
is a $w^{3-j}\in W^{3-j}$ with $Y_{w^jw^{3-j}}\neq 0$. The matrices
$A^1$ and $A^2$ are \emph{fractionally isomorphic} (we write
$A^1\simeq A^2$) if there is a fractional isomorphism from $A^1$ to
$A^2$. It is not obvious that fractional isomorphism is an equivalence
relation; this will be a consequence of the characterisation of
fractional isomorphism by equitable partitions (see Corollary~\ref{cor:fieq}).
It is also not clear that for connected matrices this notion of fractional isomorphism
coincides with the one discussed above; this is the content of Theorem~\ref{theo:equi}.

\begin{exa}
  The following five matrices are fractionally isomorphic:
  \[
  A^1=
  \begin{pmatrix}
    2
  \end{pmatrix},\quad
  A^2=
  \begin{pmatrix}
    1&1\\
    1&1
  \end{pmatrix},\quad
  A^3=
  \begin{pmatrix}
    1&1&0\\
    1&0&1\\
    0&1&1
  \end{pmatrix},\quad
  A^4=
  \begin{pmatrix}
    1&1&0&0\\
    1&0&1&0\\
    0&1&0&1\\
    0&0&1&1
  \end{pmatrix},\quad
  A^5=
  \begin{pmatrix}
    1&1&0&0&0\\
    1&1&0&0&0\\
    0&0&1&1&0\\
    0&0&1&0&1\\
    0&0&0&1&1
  \end{pmatrix}
  \]
  For all $i,j\in[5]$, a pair of matrices with all identical entries
  is a fractional isomorphism.
\end{exa}

We now relate fractional isomorphisms to the colour refinement
algorithm and equitable partitions. The standard way of running colour
refinement on two graphs is to run it on their disjoint union. We do
the same for matrices, using the direct sum instead of the
disjoint union. Let $A^1\in\RR^{V^1\times W^1}$ and
$A^2\in\RR^{V^2\times W^2}$, where $V^1,W^1,V^2,W^2$ are mutually
disjoint.  A \emph{joint partition} of $A^1,A^2$ is a partition of
$A^1\oplus A^2$, that is, a pair $(\CP,\CQ)$ of partitions of $V^1\cup
V^2$ and $W^1\cup W^2$, respectively. A joint partition $(\CP,\CQ)$ of
$A^1,A^2$ is \emph{balanced} if all $P\in\CP$ have a nonempty
intersection with both $V^1$ and $V^2$ and all $Q\in\CQ$ have a
nonempty intersection with both $W^1$ and $W^2$.
A joint partition $(\CP,\CQ)$ of $A^1,A^2$ is \emph{equitable} if it
is an equitable partition of $A^1\oplus A^2$. 
We can compute the coarsest equitable joint partition of $A^1,A^2$
using colour refinement. 

\begin{theo}
  For all matrices $A^1,A^2$, the following three statements are
  equivalent.
  \begin{enumerate}
  \item $A^1$ and $A^2$ are fractionally isomorphic.
  \item $A^1$ and $A^2$ have a balanced equitable joint partition.
  \item The coarsest equitable joint partition of $A^1$ and $A^2$ is balanced.
  \end{enumerate}
\end{theo}

\begin{proof}
  The implication (3)$\implies$(2) is trivial.

  The converse implication (2)$\implies$(3) follows from the
  observation that if some equitable joint partition is balanced, then
  the coarsest equitable joint partition is balanced as well.
  
  To prove (1)$\implies$(2), suppose that $A^1$ and $A^2$ are fractionally isomorphic. Let
  $(X,Y)$ be a fractional isomorphism. Let $(\CP_X,\CQ_Y)$
  be the corresponding partition of $A^1\oplus A^2$. By
  Theorem~\ref{theo:aut}, this partition is equitable and hence an
  equitable joint partition of $A^1,A^2$. Recall that the parts of
  the partition $\CP_X$ are the connected components of the
  undirected graph $G_X$, because $X$ is doubly stochastic. As for every
  $j\in[2]$ and $v^j\in V^j$ there is a $v^{3-j}\in V^{3-j}$ with
  $X_{v^jv^{3-j}}\neq 0$, every connected component of $G_X$ has a
  nonempty intersection with both $V^1$ and $V^2$. Similarly, every
  part of $\CQ_Y$ has a nonempty intersection with $W^1$ and
  $W^2$. Thus $(\CP_X,\CQ_Y)$ is balanced.

  To prove that (2)$\implies$(1), let $(\CP,\CQ)$ be a balanced equitable
  joint partition of $(A^1,A^2)$. Then by Theorem~\ref{theo:aut}, the
  pair $(X_\CP,Y_\CQ)$ of matrices is a fractional
  automorphism of $A^1\oplus A^2$. As $(\CP,\CQ)$ is balanced, for
  every $j\in[2]$ and $v^j\in V^j$ there is a $v^{3-j}\in V^{3-j}$
  with $(X_{\CP})_{v^jv^{3-j}}\neq 0$, and for every $j\in[2]$
  and $w^j\in W^j$ there is a $w^{3-j}\in W^{3-j}$ with
  $(Y_{\CQ})_{w^jw^{3-j}}\neq 0$. Thus
  $(X_\CP,Y_\CQ)$ is a fractional isomorphism from
  $A^1$ to $A^2$.
\end{proof}

\begin{cor}
  Let $A^1,A^2,A^3,A^4$ be matrices such that $A^1\simeq A^2$ and
  $A^3\simeq A^4$. Then $A^1\oplus A^2\simeq A^3\oplus A^4$.
\end{cor}

\begin{lem}
  Let $(\CP,\CQ)$ be the coarsest equitable joint partition of
  $A^1,A^2$.  Then for $i=1,2$, the restriction of $(\CP,\CQ)$ to
  $A^i$ is the coarsest equitable partition of $A^i$.
\end{lem}

\begin{proof}
  Let $(\CP,\CQ)$ be the coarsest equitable joint partition of $A^1,A^2$.
  Let $(\CP^i,\CQ^i)$ be the restriction of $(\CP,\CQ)$ to $A^i$, that is, 
  $\CP^i=\{ P\cap V^i \mid P\in \CP\}$ and
  $\CQ^i=\{ Q\cap W^i \mid Q\in \CQ\}$.
  We know that $(\CP^i,\CQ^i)$ is an equitable partition of $A^i$.

  Assume, for the purpose of contradiction, the partition
  $(\CP^1,\CQ^1)$ is not the coarsest one on $A^1$.  (The case for
  $A^2$ is symmetric.)  Let $(\CR^1,\CS^1)$ be an equitable partition
  of $A^1$ that is strictly coarser than $(\CP^1,\CQ^1)$.
  
  For classes $P,P'\in\CP$ we write $P\sim P'$ if $P\cap V^1$ and $P'\cap V^1$ are subsets of the same class in $\CR^1$.
  Similarly, for classes $Q,Q'\in\CQ$ we write $Q\sim Q'$ if $Q\cap W^1$ and $Q'\cap W^1$ are subsets of the same class in $\CS^1$.
  For $v\in V^1\cup V^2$, let $[v]$ denote the class of $v$ in $\CP$.
  Similarly, for $w\in W^1\cup W^2$, let $[w]$ denote the class of $w$ in $\CQ$.
  Consider the joint partition $(\CR,\CS)$ of $A^1, A^2$, 
    where $\CR$ is induced by the equivalence relation $\{(v,v') \mid [v]\sim [v'] \}$
    and $\CS$ is induced by the equivalence relation $\{(w,w') \mid [w]\sim [w'] \}$.
  The partition $(\CR,\CS)$ is equitable and coarser than $(\CP,\CQ)$, which is a contradiction.
\end{proof}

From the above Lemma, it follows immediately that fractional
isomorphism is transitive. As it is trivially reflexive and symmetric,
we obtain the following corollary.

\begin{cor}\label{cor:fieq}
  Fractional isomorphism is an equivalence relation.
\end{cor}

\begin{lem}\label{lem:equi1}
  For $j=1,2$, let $A^j\in\RR^{V^j\times W^j}$ such that
  $A^1$ is connected. Let $(\CP,\CQ)$ be a balanced equitable joint
  partition of $A^1$ and $A^2$. Then for all $P\in\CP,Q\in\CQ$,
  \[
  \frac{|V^1|}{|V^2|}=\frac{|P\cap V^1|}{|P\cap V^2|}=\frac{|Q\cap
    W^1|}{|Q\cap W^2|}=\frac{|W^1|}{|W^2|}.
  \]
\end{lem}

\begin{proof}
  As usual, for $j=1,2$ we let $F^j=F^{A^j}$.  Furthermore, for $P\in\CP$, and $Q\in\CQ$ we let
  $P^j:=P\cap V^j$ and $Q^j:=Q\cap W^j$.
  For all $P\in\CP,Q\in\CQ$, let $f_{PQ}:=F^j(v^j,Q^j)$ for some (and hence all) $j\in[2]$, $v^j\in P^j$ and
  $g_{PQ}:=F^j(P^j,w^j)$ for some (and hence all) $j\in[2]$,
  $w^j\in Q^j$. Then for $j=1,2$ we have
  \[
  |P^j|\cdot f_{PQ}=F^j(P^j,Q^j)=|Q^j|\cdot g_{PQ}.
  \]
  Since the joint partition $(\CP,\CQ)$ is balanced, we have
  $|P^j|,|Q^j|\neq 0$, and thus $f_{PQ}= 0\iff g_{PQ}=0$. Furthermore,
  if $f_{PQ}\neq 0$ then 
  \begin{equation}\label{eq:i2}
  \frac{|P^1|}{|Q^1|}=\frac{|P^2|}{|Q^2|}=\frac{g_{PQ}}{f_{PQ}}.
  \end{equation}
  Let $\CB$ be the bipartite graph with vertex set $\CP\cup\CQ$
  and (undirected) edges $PQ$ for all $P\in\CP,Q\in\CQ$ such that
  $f_{PQ}\neq0$. As the bipartite graph $B_{A^1}$ is connected, the graph $\CB$ is
  connected as well. 

  This implies that for all $P\in\CP,Q\in\CQ$ we have 
  \begin{equation}\label{eq:i3}
  \frac{|P^1|}{|P^2|}=\frac{|Q^1|}{|Q^2|}.
  \end{equation}
  As $\CB$ is connected, it suffices to prove this for $PQ\in E(\CB)$,
  and for such $PQ$ it follows immediately from 
  \eqref{eq:i2}.

  As $\sum_{P\in\CP}|P^j|=|V^j|$, this implies for all $P\in\CP$,
  \begin{equation}
    \frac{|P^1|}{|P^2|}=\frac{|V^1|}{|V^2|}.
  \end{equation}
  Similarly, for all $Q\in\CQ$,
  \begin{equation}
    \frac{|Q^1|}{|Q^2|}=\frac{|W^1|}{|W^2|}.
  \end{equation}
\end{proof}

\begin{lem}\label{lem:equi2}
  Let $A^1\in\RR^{V^1\times W^1}$ and $A^2\in\RR^{V^2\times W^2}$ be
  nonzero matrices such that there exists
  a pair
  $(X,Y)\in\RR^{V^2\times V^1}\times\RR^{W^2\times W^1}$ of
  stochastic matrices with constant column sums satisfying
  \eqref{eq:i1a} and \eqref{eq:i1b}. Then
  \[
  \frac{|V^1|}{|V^2|}=\frac{|W^1|}{|W^2|}.
  \]
\end{lem}

\begin{proof}
  Suppose for contradiction that
  $c:=\frac{|V^1|}{|V^2|}\neq\frac{|W^1|}{|W^2|}=:d$. Let us assume
  that $c>d$, the proof for the case $c<d$ is similar.

  Observe that \eqref{eq:i1a} and \eqref{eq:i1b} imply
  \[
  (XX^{\tr})^nXA^1=A^2Y(Y^{\tr}Y)^n
  \]
  for all $n\ge 0$. An easy calculation shows that the matrix
  $(XX^{\tr})$ is a matrix with constant row an column sums $1/c$, and
  the matrix $(YY^{\tr})$ is a matrix with constant row an column sums
  $1/d$. Let $C:=c(XX^{\tr})$ and $D:=d(YY^{\tr})$ and
  $M:=XA^1=A^2Y$. Then for all $n\ge 0$ we have
  \[
  \left(\frac{d}{c}\right)^n C^n M=MD^n.
  \]
  As both $C$ and $D$ are doubly stochastic and $M$ is nonzero and
  $\frac{d}{c}<1$, this leads to a contradiction.
\end{proof}

\begin{theo}\label{theo:equi}
  For $j=1,2$, let $A^j\in\RR^{V^j\times W^j}$. Suppose that
  $A^1$ is connected. Then $A^1,A^2$ are
  fractionally isomorphic if and only if there is a pair
  $(X,Y)\in\RR^{V^2\times V^1}\times\RR^{W^2\times W^1}$ of 
  stochastic matrices with constant column sums satisfying \eqref{eq:i1a} and \eqref{eq:i1b}.
\end{theo}

\begin{proof}
  To prove the forward direction, let
  $(\CP,\CQ)$ be a balanced equitable joint partition of
  $A^1,A^2$. For all $j\in[2]$, $P\in\CP$, and $Q\in\CQ$, we let
  $P^j:=P\cap V^j$ and $Q^j:=Q\cap W^j$. By Lemma~\ref{lem:equi1}, we
  have
    \[
  \frac{|V^1|}{|V^2|}=\frac{|P^1|}{|P^2|}=\frac{|Q^1|}{|Q2|}=\frac{|W^1|}{|W^2|}
  \]
  For all $P\in\CP,Q\in\CQ$, let $f_{PQ}:=F^j(v^j,Q^j)$ for some (and hence all) $j\in[2]$, $v^j\in P^j$ and
  $g_{PQ}:=F^j(P^j,w^j)$ for some (and hence all) $j\in[2]$,
  $w^j\in Q^j$. Then for $j=1,2$ we have
  \begin{equation}\label{eq:i5}
  |P^j|\cdot f_{PQ}=F^j(P^j,Q^j)=|Q^j|\cdot g_{PQ}.
  \end{equation}
 We define $X\in\RR^{V^2\times V^1}$ by $X_{v^2v^1}:=1/|P^1|$ if
  $v^2,v^1\in P$ for some $P\in\CP$ and $X_{v^2v^1}:=0$ otherwise. Similarly, we
  define $Y\in\RR^{W^2\times W^1}$ by $Y_{w^2w^1}:=1/|Q^1|$ if
  $w^2,w^1\in Q$ for some $Q\in\CQ$ and $Y_{w^2w^1}:=0$ otherwise. 
   Then for all $v^2\in V^2$,
   \[
   \sum_{v^1\in V^1}X_{v^2v^1}=\sum_{v^1\in P^1}\frac{1}{|P^1|}=1,
   \]
   where $P\in\CP$ such that $v^2\in P$.
   For all $v^1\in V^1$,
   \[
   \sum_{v^2\in V^2}X_{v^2v^1}=\sum_{v^2\in
     P^2}\frac{1}{|P^1|}=\frac{|P^2|}{|P^1|}=\frac{|V^2|}{|V^1|},
   \]
   where $P\in\CP$ such that $v^1\in P$.
   Thus $X$, and similarly $Y$, are stochastic matrices with constant column
   sums. We claim that
  $(X,Y)$ satisfies the equations \eqref{eq:i1a} and \eqref{eq:i1b}. To
  prove \eqref{eq:i1a}, let $v^2\in V^2$ and $w^1\in W^1$, and let
  $P\in\CP,Q\in\CQ$ such that $v^2\in P^2$ and $w^1\in Q^1$. Then
  \[
  (XA^1)_{v^2w^1}=\sum_{v^1\in V^1}X_{v^2v^1}A^1_{v^1w^1}=\sum_{v^1\in
    P^1}\frac{1}{|P^1|}A^1_{v^1w^1}=\frac{F^1(P^1,w^1)}{|P^1|}=\frac{g_{PQ}}{|P^1|}=\frac{f_{PQ}}{|Q^1|}=(A^2Y)_{v^2w^1},
  \]
  where $\frac{g_{PQ}}{|P^1|}=\frac{f_{PQ}}{|Q^1|}$ follows from
  \eqref{eq:i5}. Equation \eqref{eq:i1b} can be proved similarly. 

  To prove the backward direction, let $(X,Y)$ be a pair of
  stochastic matrices with constant column sums satisfying
  \eqref{eq:i1a} and \eqref{eq:i1b}. By Lemma~\ref{lem:equi2}, we have
  $|V^1|/|V^2|=|W^1|/|W^2|=:c$. Then the column sums of both $X$ and $Y$
  are $1/c$. 

  \begin{cs}
    \case1
    $c\ge1$.\\
    Let $C$ be the $V^1\times V^1$ diagonal
  matrix with all diagonal entries $1-1/c$, and let $D$ be the $W^1\times W^1$ diagonal
  matrix with all diagonal entries $1-1/c$. Then 
  \[
  \left(\;
    \begin{pmatrix}
      C&X^{\tr}\\X&0
    \end{pmatrix},\;
    \begin{pmatrix}
      D&Y^{\tr}\\Y&0
    \end{pmatrix}\;
  \right)
  \]
  is a fractional automorphism from $A^1$ to $A^2$.
  \case2
  $c<1$.\\
  In this case, let $C$ be the $V^2\times V^2$ diagonal
  matrix with all diagonal entries $1-c$, and let $D$ be the $W^2\times W^2$ diagonal
  matrix with all diagonal entries $1-c$. Then 
  \[
  \left(\;
    \begin{pmatrix}
      0&cX^{\tr}\\cX&C
    \end{pmatrix},\;
    \begin{pmatrix}
      0&cY^{\tr}\\cY&D
    \end{pmatrix}\;
  \right)
  \]
  is a fractional automorphism from $A^1$ to $A^2$.
  \qedhere
  \end{cs}
\end{proof}

The following example shows that the condition in
Theorem~\ref{theo:equi}, that $A^1$ (or equivalently $A^2$) be
connected is necessary, even for $0$-$1$-matrices.

  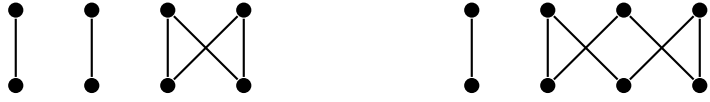
\begin{figure}
    \centering
    \begin{tikzpicture}
      [
      every node/.style={circle, fill, minimum height=2mm,inner sep=0pt},
      thick
      ]
      \begin{scope}
        \path (0,0) node (v1) {} (1,0) node (v2) {} (2,0) node (v3) {}
        (3,0) node (v4) {}
        (0,1) node (w1) {} (1,1) node (w2) {} (2,1) node (w3) {} (3,1) node (w4) {}
        ;
        \draw (v1) edge (w1) (v2) edge (w2) (v3) edge (w3) edge (w4)
        (v4) edge (w3) edge (w4);
        \end{scope}
        
        \begin{scope}[xshift=6cm]
        \path (0,0) node (v1) {} (1,0) node (v2) {} (2,0) node (v3) {}
        (3,0) node (v4) {}
        (0,1) node (w1) {} (1,1) node (w2) {} (2,1) node (w3) {} (3,1) node (w4) {}
        ;
        \draw (v1) edge (w1) (v2) edge (w2) edge (w3) (v3) edge (w2) edge (w4)
        (v4) edge (w3) edge (w4);
        \end{scope}
    \end{tikzpicture}
    \caption{Bipartite graphs of the matrices in Example~\ref{exa:equi}}
    \label{fig:equi}
  \end{figure}

\begin{exa}\label{exa:equi}
  The matrices
  \[
  A^1=
  \begin{pmatrix}
    1&0&0&0\\
    0&1&0&0\\
    0&0&1&1\\
    0&0&1&1
  \end{pmatrix},
  \qquad
  A^2=
  \begin{pmatrix}
    1&0&0&0\\
    0&1&1&0\\
    0&1&0&1\\
    0&0&1&1
  \end{pmatrix},
  \]
  whose bipartite graphs are displayed in Figure~\ref{fig:equi}, are
  easily seen to be fractionally isomorphic, but there is no pair
  $(X,Y)$ of doubly stochastic matrices satisfying \eqref{eq:i1a} and \eqref{eq:i1b}. We
  leave it to the reader to prove these claims.
\end{exa}

%%%%%%%%%%%%%%%%%%%%%%%%%%%%%%%%%%%%%%%%%%%%%%%%%%%%%%%%%%%%
\section{Factor Matrices and Partition Equivalence}

For our applications, fractional isomorphism is an equivalence
relation that is still too fine. In this section, we introduce a
coarser equivalence relation that we will call \emph{partition equivalence}.
For the applications in the next section, it
will be helpful to develop partition equivalence for matrices with
entries from $\ORR=\RR\cup\{\infty\}$. 

\subsection{Partition Matrices and Factor Matrices}
A \emph{partition
  matrix} is a $0$-$1$ matrix that has exactly one $1$-entry in each
row and at least one $1$-entry in each column. We
usually denote partition matrices by $\Pi$ or $C,D$. With each partition
matrix $\Pi\subseteq\{0,1\}^{V\times T}$ we associate a partition
$\{P_t\mid t\in T\}$ of $V$ into parts $P_t=\{v\in V\mid
\Pi_{vt}=1\}$. Conversely, with every partition $\CP$ of $V$ we
associate the partition matrix $\Pi_{\CP}\in\{0,1\}^{V\times\CP}$
defined by $(\Pi_\CP)_{vP}=1\iff v\in P$, for all $v\in V$ and $P\in\CP$. 

Note that partition matrices are stochastic, but, in general, not
doubly stochastic. (The only doubly stochastic partition matrices are
the permutation matrices.) For every partition matrix
$\Pi\in\RR^{V\times T}$, we define its \emph{scaled transpose} to be the
matrix $\Pi^{\str}\in\RR^{T\times V}$
with entries
\[
\Pi^{\str}_{tv}:=\frac{\Pi_{vt}}{\sum_{v'\in V}\Pi_{v't}}.
\]
Then $\Pi^{\str}$ is the transpose of $\Pi$ scaled to a stochastic matrix. 
Observe that the matrix $\Pi\Pi^{\str}\in\RR^{V\times V}$ is symmetric and doubly stochastic. Indeed, if
$\Pi=\Pi_{\CP}$ for a partition $\CP$ of $V$, then 
\begin{equation}\label{eq:pe0a}
(\Pi\Pi^{\str})_{vv'}=
\begin{cases}
  1/|P|&\text{if }v,v'\in P\text{ for some }P\in\CP,\\
  0&\text{otherwise}.
\end{cases}
\end{equation}
This is precisely the matrix $X_{\CP}$ defined on
page~\pageref{page:XP}. Thus we obtain the following corollary of
Theorem~\ref{theo:aut}.

\begin{cor}\label{cor:pe-fi}
  Let $(\CP,\CQ)$ be an equitable partition of a matrix
  $A\in\RR^{V\times W}$, and let
  $C:=\Pi_{\CP}$ and $D:=\Pi_{\CQ}$. Then $(CC^{\str},DD^{\str})$ is a fractional
  automorphism of $A$.
\end{cor}

A \emph{factor matrix} of a matrix $A\in\ORR^{V\times W}$ is a
matrix 
\[
B=\Pi_{\CP}^{\str}A\Pi_{\CQ}\in\ORR^{\CP\times\CQ},
\]
where $(\CP,\CQ)$ is an equitable partition of $A$. The asymmetry in
the definition (multiplying with $\Pi_{\CP}^{\str}$ rather than
$\Pi_{\CP}^{\tr}$) may seem strange first, but turns out to be
necessary in several places. An immediate advantage of it is that we
multiply with stochastic matrices from both sides.
Note that for all
$P\in\CP,Q\in\CQ$,
\begin{equation}\label{eq:pe1}
B_{PQ}=\frac{1}{|P|}F^A(P,Q)=F^A(v,Q)
\end{equation}
for some (and hence for all) $v\in P$. (This is the number we
sometimes denote by $f_{PQ}$.)  We will see that factor matrices still
carry all information about a matrix necessary to solve systems of
linear equations and linear programs. 

As the dimensions of $B$ are determined by
the number of classes of the partition, there is a
unique smallest factor matrix
\[
[A]:=\Pi_{\CP_\infty}^{\str}A\Pi_{\CQ_\infty},
\]
where, as usual, $(\CP_\infty,\CQ_\infty)$ denotes the coarsest
equitable partition of $A$. We call $[A]$ the \emph{core factor} of
$A$.  Theorem~\ref{theo:alg} implies that we can compute the core
factor in quasilinear time.

\begin{cor}
  There is an algorithm that, given a sparse representation of a
  matrix $A$, computes the core factor $[A]$ in time
  $O((n+m)\log n)$.
\end{cor}

\subsection{Partition Equivalence}
We define the relation $\approx$ on the class of all matrices by letting $A^1\approx
A^2$ if there are factor matrices $B^1$ of $A^1$ and $B^2$ of $A^2$
such that $B^1$ and $B^2$ are isomorphic. 
Observe that for every matrix
$A\in\ORR^{V\times W}$,
\begin{equation}\label{eq:ep2}
A\approx [A],
\end{equation}
because $I[A]I'$, where $I,I'$ denote the identity
matrices of the right dimensions.

Moreover, for all matrices $A^1\in\ORR^{V^1\times W^1}$ and
$A^2\in\ORR^{V^2\times W^2}$ we have
\begin{equation}\label{eq:ep3}
A^1\simeq A^2\implies A^1\approx A^2.
\end{equation}
To see this, suppose that $A^1\simeq A^2$ and let $(\CP,\CQ)$ be a balanced
equitable joint partition of $A^1$ and $A^2$. For $j=1,2$, let
$\CP^j:=\{P\cap V^j\mid P\in\CP\}$ and $\CQ^j:=\{Q\cap W^j\mid
Q\in\CQ\}$. Then $(\CP^j,\CQ^j)$ is an equitable partition of $A^j$. Let $C^j:=\Pi_{\CP^j}$ and $D^j:=\Pi_{\CQ^j}$ and
\[
B^j:=(C^j)^{\str}A^jD^j.
\]
We claim that $B^1\cong B^2$ via the isomorphism $P^1\mapsto P^2$,
$Q^1\mapsto Q^2$. To see this, let $P\in\CP,Q\in\CQ$. Then by
\eqref{eq:pe1} and \eqref{eq:e1}
\[
B^1_{P^1Q^1}=F^{A^1}(v^1,Q^1)=F^{A^1\oplus A^2}(v^1,Q) =F^{A^1\oplus A^2}(v^2,Q)=F^{A^2}(v^2,Q^2)=B^2_{P^2Q^2}
\]
for some (and hence for all) $v^1\in P^1,v^2\in P^2$. 

Maybe surprisingly, the relation $\approx$ is \emph{not} an
equivalence relation. It is obviously reflexive and symmetric, but the
next example shows that it is not transitive.

\begin{exa}\label{exa:icf}
  Consider the matrix $\tilde A$ of Example~\ref{exa:intro} and the
  factor matrices $[\tilde A]$ and $[[\tilde A]]$. 
    By \eqref{eq:ep2}, we have $\tilde A\approx[\tilde A]$ and
  $[\tilde A]\approx[[\tilde A]]$. However, we clearly have $\tilde
  A\not\approx [[\tilde A]]$,
  because there is no factor of $\tilde A$ and thus no matrix $B$ with
  $\tilde A\approx B$ that is smaller
  than the core factor.
\end{exa}

We let $\approx^*$ be the transitive closure of $\approx$ and call two
matrices $A^1,A^2$ \emph{partition equivalent} if $A^1\approx^* A^2$.

Let $A\in\ORR^{V\times W}$. We let $[A]_0:=A$ and $[A]_{i+1}:=[[A]_i]$
for every $i\ge 0$. Then there is an $i\le|V|+|W|$ such that
$[A]_i=[A]_{i+1}$. We denote $[A]_i$ by $\llbracket A\rrbracket$ and
call it the \emph{iterated core factor} of $A$. Observe that
\eqref{eq:ep2} implies that 
\begin{equation}
  \label{eq:ep4}
  A\approx^*\llbracket A\rrbracket.
\end{equation}

The following example shows that $A^1\approx^* A^2$ does not
necessarily imply 
$\llbracket A^1\rrbracket\cong\llbracket A^2\rrbracket$ and that
$\llbracket A\rrbracket$ is not necessarily the smallest matrix
partition equivalent to $A$. This is unfortunate, because it leaves us without an
efficient way of deciding partition equivalence. 

\begin{exa}\label{exa:cf}
  Consider the matrices
  \[
  A^1=\left(
  \begin{array}{cccc|cc}
    0&0&1&0&1&0\\
    0&0&0&1&1&0\\
    1&0&0&0&0&1\\
    0&1&0&0&0&1\\
    \hline
    1&1&0&0&0&0\\
    0&0&1&1&0&0
  \end{array}\right)
  \quad\text{and}\quad
  A^2=
  \begin{pmatrix}
    1&1\\
    2&0
  \end{pmatrix}
  \]
  Suppose that for $j=1,2$ the rows and columns of the matrix $A^j$ are indexed by
  $v_i^j,w_i^j$. 

  The matrices $A^1$ and $A^2$ are partition equivalent: for the
  equitable partition (indicated by the lines in $A^1$)
  \[
  (\CP,\CQ):=\Big(
  \big\{ \{v_1^1,v_2^1,v_3^1,v_4^1\},\{v_5^1,v_6^1\}\big\},
  \big\{ \{w_1^1,w_2^1,w_3^1,w_4^1\},\{w_5^1,w_6^1\}\big\}
  \Big)
  \]
  of $A^1$ we have
  \[
  \Pi^{\str}_\CP A^1\Pi_\CQ=
  \begin{pmatrix}
    \frac{1}{4}&\frac{1}{4}&\frac{1}{4}&\frac{1}{4}&0&0\\
    0&0&0&0&\frac{1}{2}&\frac{1}{2}
  \end{pmatrix}
    \begin{pmatrix}
    0&0&1&0&1&0\\
    0&0&0&1&1&0\\
    1&0&0&0&0&1\\
    0&1&0&0&0&1\\
    1&1&0&0&0&0\\
    0&0&1&1&0&0
  \end{pmatrix}
  \begin{pmatrix}
    1&0\\
    1&0\\
    1&0\\
    1&0\\
    0&1\\
    0&1
  \end{pmatrix}
  =  \begin{pmatrix}
    1&1\\
    2&0
  \end{pmatrix}
  \cong A^2.
  \]
  Thus $A^1\approx A^2$. 
  
  The coarsest equitable partition of $A^1$ is the trivial partition
  $\Big(\big\{\{v_1^1,\ldots,v_6^1\}\big\},\big\{\{w_1^1,\ldots,w_6^1\}\big\}\Big)$,
  and thus
  \[
  \llbracket A^1\rrbracket=[A^1]=(2).
  \]
  The coarsest equitable partition of $A^2$ is the identity partition
  $\Big(\big\{\{v_1^2\},\{v_2^2\}\big\},\big\{\{w_1^2\},\{w_2^2\}\big\}\Big)$,
  which yields
  \[
  \llbracket A^2\rrbracket=A^2.
  \]
  Thus $\llbracket A^1\rrbracket\not\cong \llbracket A^2\rrbracket$.

  Note that this also means that the smallest matrix partition
  equivalent to $A^2$ is not $\llbracket A^2\rrbracket$, as one might
  have expected, but $\llbracket A^1\rrbracket$.
\end{exa}

It remains an open
question whether partition equivalence is decidable, or even
decidable in polynomial time. Note, however, that we can compute
$\llbracket A\rrbracket$ from $A$ in time $O(n(n+m)\log n)$. It is
conceivable that this can be improved to $O((n+m)\log n)$, but this
remains open as well.

%%%%%%%%%%%%%%%%%%%%%%%%%%%%%%%%%%%%%%%%%%%%%%%%%%%%%%%%%%%%
\section{Reducing the Dimension of a Linear Program}
In this section, we will apply our theory of fractional automorphisms
and partition equivalence to solving systems of linear equations and
linear programs.
Let $A\in\RR^{V\times W}$, $b\in\RR^V$, $c\in\RR^W$,
and $x=(x_w\mid w\in W)$. We consider the system \eqref{EAb}
\begin{equation}
  \label{EAb}
  \tag{$E_{A,b}$}
  Ax=b
\end{equation}
of linear equations, the linear program \eqref{LAbc} in standard form
\begin{equation}
  \label{LAbc}
  \tag{$L_{A,b,c}$}
  \begin{array}{rl}
    \text{min }&c^{\tr} x\\
\text{subject to }&Ax=b,\;x\ge 0,
  \end{array}
\end{equation}
and the linear program \eqref{DAbc} in dual form
\begin{equation}
  \label{DAbc}
  \tag{$D_{A,b,c}$}
  \begin{array}{rl}
    \text{max }&c^{\tr} x\\
\text{subject to }&Ax\le b.
  \end{array}
\end{equation}
Our results actually extend to arbitrary linear programs, but we focus
on these for the ease of presentation.

We need to take the vectors $b$ and $c$ into
account. Let $A,b,c$ be as above. We define a matrix $\tilde A=\tilde
A(A,b,c)\in\ORR^{(V\cup\{v_\infty\})\times (W\cup\{w_\infty\})}$, where we assume
$v_\infty,w_\infty\not\in V\cup W$ and $v_\infty\neq w_\infty$, by
\[
\tilde A_{vw}:=
\begin{cases}
  A_{vw}&\text{if }v\in V,w\in W,\\
  c_w&\text{if }v=v_\infty,w\in W,\\
  b_v&\text{if }v\in V,w=w_\infty,\\
  \infty&\text{if }v=v_\infty,w=w_\infty
\end{cases}
\]
(see Example~\ref{exa:intro}).
As $A$ is a real matrix (that does not contain $\infty$ as an entry), every equitable partition
$(\tilde\CP,\tilde\CQ)$ of $\tilde A$ contains $\{v_\infty\}$ and $\{w_\infty\}$
as separate classes. If we let $\CP:=\tilde\CP\setminus\{v_\infty\}$ and
$\CQ:=\tilde\CQ\setminus\{w_\infty\}$, then $(\CP,\CQ)$ is an equitable
partition of $A$ satisfying
\begin{align}
  \label{eq:a1}
  b_v&=b_{v'}&\text{for all }P\in\CP,v,v'\in P,\\
  \label{eq:a2}
  c_w&=c_{w'}&\text{for all }Q\in\CQ,w,w'\in Q.
\end{align}
Furthermore, if $(\tilde\CP,\tilde\CQ)$ is the coarsest equitable
partition of $\tilde A$ then $(\CP,\CQ)$ is the coarsest equitable
partition of $A$ that satisfies \eqref{eq:a1} and \eqref{eq:a2}.

\begin{lem}[Reduction Lemma]\label{lem:red}
  Let $A,b,c,\tilde A$ as above. Let $(\tilde\CP,\tilde\CQ)$ an
  equitable partition of $\tilde A$ and
  $\CP:=\tilde\CP\setminus\{v_\infty\}$,
  $\CQ:=\tilde\CQ\setminus\{w_\infty\}$. Furthermore, let $C:=\Pi_\CP$ and
  $D:=\Pi_{\CQ}$ and $A':=C^{\str}AD$, $b':=C^{\str}b$, and $c':=D^{\tr}c$.
  \begin{enumerate}
  \item If $x\in\RR^W$ is a solution to \eqref{EAb} or a feasible
    solution to \eqref{LAbc} or \eqref{DAbc}, then $x':=D^{\str}x$ is a
    solution to \textup($E_{A',b'}$\textup) or a feasible solution to
    \textup($L_{A',b',c'}$\textup) or \textup($D_{A',b',c'}$\textup),
    respectively.

    Moreover, if $x$ is an optimal solution to \eqref{LAbc} or
    \eqref{DAbc}, then $x'$ is an optimal solution to
    \textup($L_{A',b',c'}$\textup) or \textup($D_{A',b',c'}$\textup),
    respectively.

  \item If $x'\in\RR^{\CQ}$ is a solution to ($E_{A',b'}$) or a
    feasible solution to \textup($L_{A',b',c}$\textup) or
    \textup($D_{A',b',c'}$\textup), then $x:=Dx'$ is a solution to
    \eqref{EAb} or a feasible solution to \eqref{LAbc} or
    \eqref{DAbc}, respectively.

    Moreover, if $x'$ is an optimal solution to
    \textup($L_{A',b',c}$\textup) or \textup($D_{A',b',c'}$\textup),
    then $x$ is an optimal solution to \eqref{LAbc} or \eqref{DAbc},
    respectively.
  \end{enumerate}
\end{lem}

\begin{proof}
  We only prove this for the linear programs \eqref{LAbc} and
  ($L_{A',b',c}$) in standard form; the proofs for systems of linear
  equations and linear programs in dual form are similar.

  Observe first that $C^{\str}CC^{\str}=C^{\str}$. To see this, recall from
  \eqref{eq:pe0a} that $(CC^{\str})_{v'v}=\frac{1}{|P|}$ if
  $v,v'\in P$ for some $P\in\CP$ and $(CC^{\str})_{v'v}=0$
  otherwise. Let $v\in V$ and
  $P\in\CP$. Then
  \[
  (C^{\str}CC^{\str})_{Pv}=\sum_{v'\in V}C^{\str}_{Pv'}(C^{\str}C)_{v'v}=\sum_{v'\in P}\frac{1}{|P|}(C^{\str}C)_{v'v}=\left\{
  \begin{array}{ll}
    \sum_{v'\in P}\left(\frac{1}{|P|}\right)^2=\frac{1}{|P|}&\text{if
    }v\in P,\\
    0&\text{otherwise}.
  \end{array}
  \right\}=C^{\str}_{Pv}.
  \]
  Similarly, $DD^{\str}D=D$.

  To prove (1), let $x\in\RR^W$ be a feasible solution to \eqref{LAbc}
  and $x':=D^{\str}x\in\RR^{\CQ}$. Then $x'\ge 0$ because $x\ge0$ and $D^{\str}$
  is nonnegative. Furthermore,
  \[
  A'x'=C^{\str}ADD^{\str}x\overset{(a)}{=}C^{\str}CC^{\str}Ax\overset{(b)}{=}C^{\str}b=b'.
  \]
  Here (a) holds because $(CC^{\str},DD^{\str})$ is a fractional automorphism of
  $A$ and (b) holds because $C^{\str}CC^{\str}=C^{\str}$ and $Ax=b$. Thus $x'$ is a feasible
  solution to ($L_{A',b',c}$).

  Before we prove the second assertion of (1) regarding optimal solutions,
  we prove the first assertion of (2).
  Let $x'\in\RR^{\CQ}$ be a feasible solution to ($L_{A',b',c}$)
  and $x:=Dx'\in\RR^W$. Then $x\ge 0$ because $x'\ge0$ and $D$
  is nonnegative. Furthermore,
  \[
  Ax=ADx'\overset{(c)}=ADD^{\str}Dx'\overset{(d)}{=}CC^{\str}ADx'=CA'x'=Cb'=CC^{\str}b\overset{(e)}{=}b.
  \]
  Here (c) holds, because $DD^{\str}D=D$, and (d) holds, because $(CC^{\str},DD^{\str})$ is a fractional automorphism of
  $A$. To prove (e), let $v\in V$, and let $P\in\CP$ such that $v\in
  P$. By \eqref{eq:a1}, we have 
  \[
  b_v=\frac{1}{|P|}\sum_{v'\in P}b_{v'}= (CC^{\str}b)_v.
  \]
  Thus $x$ is a feasible solution to \eqref{LAbc}.

  It remains to prove the two assertions about optimal solutions. Suppose
  first that $x\in\RR^W$ is an optimal solution to \eqref{LAbc}, and
  let $x':=D^{\str}x$. Then $x'$ is a feasible solution to
  ($L_{A',b',c}$). We claim that it is optimal. Let $y'$ be another
  feasible solution to ($L_{A',b',c}$). We shall prove that
  $(c')^{\tr}x'\le (c')^{\tr}y'$. Let $y=Dy'$. Then $y$ is a feasible solution
  to \eqref{LAbc}, and thus $c^{\tr}x\le c^{\tr}y$ by the optimality of $x$.
  Thus
  \[
  (c')^{\mathrm{t}}x'=c^{\tr}DD^{\str}x\overset{(f)}{=}c^{\tr}x\le c^{\tr}y=c^{\tr}Dy'=(c')^{\tr}y'.
  \]
  Here (f) holds because $c^{\tr}DD^{\str}=c^{\tr}$. To see this, let $w\in W$ and
  $Q\in\CQ$ such that $w\in Q$. Then by \eqref{eq:a2},
  \[
  c^{\tr}_w=\frac{1}{|Q|}\sum_{w'\in Q}c_{w'}=(c^{\tr}DD^{\str})_{w}.
  \]
  Suppose conversely that $x'\in\RR^{\CQ}$ is an optimal solution to ($L_{A',b',c}$)
  and let $x:=Dx'$. Then $x$ is a feasible solution to
  \eqref{LAbc}. Let $y$ be another feasible solution. Then $y':=D^{\str}y$ is a
  feasible solution to ($L_{A',b',c}$), and by the optimality of $x'$
  we have $(c')^{\tr}x'\le (c')^{\tr}y'$. Thus
  \[
  c^{\tr}x\overset{(f)}{=}c^{\tr}DD^{\str}x=(c')^{\tr}x'\le(c')^{\tr}y'=c^{\tr}DD^{\str}y\overset{(f)}{=}c^{\tr}y.
  \]
  The two equations marked (f) hold, because $c^{\tr}DD^{\str}=c^{\tr}$, as
  we have seen above.
\end{proof}

Note that if we apply the reduction lemma to a system \eqref{EAb} of
linear equations, then the vector $c$ is irrelevant, and we can simply
let $c=0$.

For simplicity, in the following we state our results only for linear
programs in standard form. The corresponding results for linear
programs in dual form or systems of linear equations hold as well.

\begin{theo}\label{theo:red}
  For $j=1,2$, let $A^j\in\RR^{V^j\times W^j}$ and $b^j\in\RR^{V^j}$
  and $c^j\in\RR^{W^j}$ and $\tilde A^j:=\tilde
  A(A^j,b^j,c^j)$. Suppose that
  \[
  \tilde A^1\approx^*\tilde A^2.
  \]
  Then for $j=1,2$ there is a matrix $M^j\in\RR^{W^{3-j}\times W^j}$ such that for all $x\in\RR^{W^j}$, if $x$ is a feasible solution to
    \textup($L_{A^j,b^j,c^j}$\textup) 
    then $M^jx$ is a feasible solution to
    \textup($L_{A^{3-j},b^{3-j},c^{3-j}}$\textup). 

    Furthermore,   if $x$ is an optimal solution to
    \textup($L_{A^j,b^j,c^j}$\textup) 
    then $M^jx$ is an optimal solution to
    \textup($L_{A^{3-j},b^{3-j},c^{3-j}}$\textup). 
\end{theo}

\begin{proof}
  Assume first that $\tilde A^1\approx \tilde A^2$. For $j=1,2$, let
  $(\tilde\CP^j,\tilde\CQ^j)$ be an equitable partition of $A^j$ such
  that 
  \[
  \tilde B^1:=\Pi^{\str}_{\tilde\CP^1}\tilde A^1
  \Pi^{\str}_{\tilde\CQ^1}\cong 
  \Pi^{\str}_{\tilde\CP^2}\tilde A^2
  \Pi^{\str}_{\tilde\CQ^2}=:\tilde B^2.
  \]
  Let
  $\tilde\pi:\tilde\CP^1\to\tilde\CP^2$ and
  $\tilde\rho:\tilde\CQ^1\to\tilde\CQ^2$ be bijections such that $\tilde
  B^1_{PQ}=\tilde B^2_{\tilde\pi(P)\tilde\rho(Q)}$ for all
  $P\in\tilde\CP^1,Q\in\tilde\CQ^1$. Let $v_\infty^j,w_\infty^j$ be the indices
  of the row and column of $\tilde A^j$ such that $\tilde
  A^j_{v_\infty^jw_\infty^j}=\infty$. Recall that $P_\infty^j:=\{v_\infty^j\}$ and $Q_\infty^j:=\{w_\infty^j\}$
  are classes in $\tilde\CP^j,\tilde\CQ^j$, respectively.  Observe that $\tilde
  B^j_{P_\infty^jQ_\infty^j}=\infty$ and that all other entries of $\tilde B^j$
  are real. Thus $\tilde\pi(P_\infty^1)=P_\infty^2$
  and $\tilde\rho(Q_\infty^1)=Q_\infty^2$. 

  Let
  $\CP^j:=\tilde\CP^j\setminus\{P_\infty^j\}$ and
  $\CQ^j:=\tilde\CP^j\setminus\{P_\infty^j\}$ and $C^j:=\Pi_{\CP^j},D^j:=\Pi_{\CQ^j}$. Let
  $B^j=(C^j)^{\str}A^jD^j\in\RR^{\CP^j\times\CQ^j}$. Then $B^j_{PQ}=\tilde
  B^j_{PQ}$ for all $P\in\CP^j,Q\in\CQ^j$. Let $d^j:=(C^j)^{\str}b^j$. Then
  $d^j_P=\tilde B^j_{PQ_\infty^j}$ for all $P\in\CP^j$. Let
  $e^j:=(D^j)^\tr c^j$. Then
  $e^j_Q=\tilde B^j_{P_\infty^jQ}$ for all $Q\in\CQ^j$. 

  Now let $\pi:\CP^1\to\CP^2$ be the restriction of $\tilde\pi$ to $\CP$, and let
  $\rho:\CQ^1\to\CQ^2$ be the restriction of $\tilde\rho$ to
  $\CQ$. Then for all $P\in\CP^1,Q\in\CQ^1$ we have
  $B^1_{PQ}=B^2_{\pi(P)\rho(Q)}$ and $d^1_P=d^2_{\pi(P)}$ and
  $e^1_Q=e^2_{\rho(Q)}$. Let $X\in\RR^{\CP^1\times\CP^2}$ and
  $Y\in\RR^{\CQ^1\times\CQ^2}$ be permutation matrices corresponding
  to the bijections $\pi,\rho$, respectively. Then $X^{\tr}B^1Y=B^2$ and
  $X^{\tr}b^1=b^2$ and $Y^{\tr}c^1=c^2$. 

  Now let $M^1:=D^2Y^{\tr}(C^1)^{\str}$ and $M^2:=D^1Y(C^2)^{\str}$. It follows from
  the Reduction Lemma~\ref{lem:red} that these matrices satisfy the
  conditions of the lemma.

  Let us now consider the general case $A^1\approx^*A^2$. Then there
  is a sequence $\tilde A^1=B^1,B^2,\ldots,B^m=\tilde A^2$ of matrices
  such that $B^i\approx B^{i+1}$ for all $i\in[m-1]$. The reason the
  claim does not follow immediately from the claim for matrices
  $A^1\approx A^2$ by multiplying the chain of matrices $M^j$ is that the
  intermediate matrices $B^i$ may not be of the form $\tilde
  A(A,b,c)$. Observe that a matrix $B$ is of this form if and only if it has exactly one $\infty$-entry. As
  $\infty$ appears in $A^1$, it is clear that all $B^i$ have at least
  one $\infty$-entry. But they may have more than one. We can handle
  this by collapsing all $\infty$-entries to a single one. To make
  this precise, consider a matrix $B\in\ORR^{V\times W}$. Let
  $V_\infty\subseteq V$ be the set of all indices of rows with at least one
  $\infty$-entry, and let $W_\infty\subseteq W$ be the set of all indices
  of columns with at least one $\infty$-entry. Observe that every
  equitable partition of $B$ refines the partition $(\{V_\infty,V\setminus
  V_\infty\},\{W_\infty,W\setminus W_\infty\})$.  Let $v_\infty\in V_\infty$ and $w_\infty\in W_\infty$
  be arbitrary,
  and let $\overline V:=(V\setminus V_\infty)\cup\{v_\infty\}$ and $\overline
  W:=(W\setminus W_\infty)\cup\{w_\infty\}$. We define the matrix $\overline
  B\in\RR^{\overline V\times\overline W}$ by
  \[
  \overline B_{vw}:=
  \begin{cases}
    B_{vw}&\text{if }v\in V\setminus V_\infty,w\in W\setminus W_\infty,\\
    \frac{1}{|V_\infty|}\sum_{v'\in V_\infty}B_{v'w}&\text{if }v=v_\infty,w\in
    W\setminus W_\infty,\\
    \sum_{w'\in W_\infty}B_{vw'}&\text{if }v\in V\setminus V_\infty,w=w_\infty,\\
    \infty&\text{if }v=v_\infty,w=w_\infty.
  \end{cases}
  \]
  It is not hard to prove that for all matrices $B,B'$,
  \[
  B\approx B'\implies \overline B\approx\overline B'.
  \]
  Thus, coming back to the sequence $\tilde
  A^1=B^1,B^2,\ldots,B^m=\tilde A^2$, we have
  \[
  A^1=\overline A^1=\overline B^1\approx\overline
  B^2\approx\ldots\approx \overline B^m=\overline A^2=A^2.
  \]
  We have $A^j=\overline A^j$, because $A^j$ only has one
  $\infty$-entry. The assertion of the lemma follows.
\end{proof}

Example~\ref{exa:intro} illustrates how the theorem can be
applied. The matrix $D$ in \eqref{eq:exa-intro} is the product of the
partition matrices corresponding to the coarsest equitable partitions
of $\tilde A$ and $[\tilde A]$.

\subsection{Implementation}\label{sec:imp}
Note that Theorem~\ref{theo:red} is not algorithmic, because we do not
know how to decide partition equivalence. Fortunately, this is not
a problem for the main application, where we only apply the
Reduction Lemma~\ref{lem:red} once to the coarsest equitable
partition. We believe that the gain we may have by searching for a
smaller partition equivalent matrix than the core factor, for example
the iterated core factor, is almost always outweighed by the additional
time spent to find such a matrix. But we have not yet conducted any
systematic experiments in this direction yet.

Let us
briefly describe our implementation. 
We are given $A\in\RR^{V\times W}$, $b\in\RR^V$, $c\in\RR^W$ and want
to solve the linear program \eqref{LAbc}. To apply the Reduction
Lemma, instead of computing the coarsest equitable partition of the
matrix $\tilde A(A,b,c)$, we directly compute the coarsest equitable
partition $(\CP,\CQ)$ of $A$ that refines an initial partition
$(\CP_0,\CQ_0)$ depending on the vectors $b$ and $c$: $\CP_0$ is the
partition of $V$ where $v$ and $v'$ are in the same class if
$b_v=b_{v'}$, and $\CQ_0$ is defined similarly from $c$.
(Then $(\CP\cup\{\{v_\infty\}\},\CQ\cup\{\{w_\infty\}\})$ is the coarsest equitable
partition of $\tilde A$.) We compute $(\CP,\CQ)$ using colour
refinement starting from the initial partition $(\CP_0,\CQ_0)$. 

Our colour refinement implementation is based on the algorithm
described in Section~\ref{sec:alg}.

\subsection{Comparison with Symmetry Reduction}\label{sec:symm}
B\"odi, Grundh\"ofer and Herr~\cite{bodgruher10} proposed the
following method of symmetry reduction for linear programs. They define an \emph{automorphism} of \eqref{LAbc} to be a pair
$(X,Y)$ of permutation matrices such that $XA=AY$ and $Xb=b$ and
$Y^\tr c=c$. Automorphisms have an obvious group structure; let
$\Aut(L)$ denote the group of all automorphisms. B\"odi et al.\
observe that for every feasible solution $x$ to \eqref{LAbc},
\[
x'=\frac{1}{|\Aut(L)|}\sum_{(X,Y)\in\Aut(L)}Yx
\]
is a feasible solution as well, and if $x$ is an optimal
solution then $x'$ is an optimal solution. They argue that $x'$ is in
the intersection $E$ of the $1$-eigenspaces of all matrices $Y$ such that
$(X,Y)\in\Aut(L)$ for some $X$. If there are many automorphisms, the
dimension of $E$ can be expected to be much smaller than $n$, and thus
we can reduce the number of variables of the linear program by
projecting to $E$.

To see that this method of symmetry reduction is
subsumed by our Reduction Lemma, observe that the pair $(X,Y)$ of matrices defined by
\begin{equation}\label{eq:bhj}
X:=\frac{1}{|\Aut(L)|}\sum_{(X',Y')\in\Aut(L)}X'
\quad\text{and}\quad
Y:=\frac{1}{|\Aut(L)|}\sum_{(X',Y')\in\Aut(L)}Y'
\end{equation}
is a fractional automorphism of $A$ with $Xb=b$ and $Y^\tr c=c$, and
thus it yields a fractional automorphism of $\tilde A$. By Theorem~\ref{theo:aut},
$(\CP_X,\CQ_Y)$ is an equitable partition of $A$.
The dimension of the $E$ is equal to the rank
of $Y$,  which is at least $|\CQ_Y|$ and thus at least $|\CQ|$ for the
coarsest equitable partition $(\CP,\CQ)$ of $A$ satisfying
\eqref{eq:a1} and \eqref{eq:a2}. Thus the dimension of the linear
program we obtain via the Reduction Lemma is at most that of the linear
program that B\"odi et al.\ project to. The additional benefit of our
method is that colour refinement is much more efficient than computing
the automorphism group of a linear program. (Our experiments,
described in the next section, show that
this last point is what makes our method significantly more efficient
in practice.)

%%%%%%%%%%%%%%%%%%%%%%%%%%%%%%%%%%%%%%%%%%%%%%%%%%%%%%%%%%%%
\section{Computational Evaluation}

Our intention here is to investigate the computational benefits of
colour refinement for solving linear programs in the presence of
symmetries. To this aim, we realised our colour refinement based on
the Saucy~\cite{saucy12}, where the unweighted version is already
implemented as a preprocessing heuristic for automorphism group
computation. We modified the code to return the colour classes after
preprocessing and not proceed with the actual automorphism
search. From the colour classes we computed the reduced LPs according
to Lemma~\ref{lem:red}. We used CVXOPT (\url{http://cvxopt.org/}) for
solving the original and reduced linear programs. We report on the
dimensions of the linear programs and on the running times when
solving the original linear programs (without compression) as well as
the reduced ones using colour refinement. We additionally compare the
results to the symmetry reduction approach due to B\"odi et
al.~\cite{bodgruher10} described in Section~\ref{sec:symm} (which we
also implemented using Saucy). All experiments were conducted on a
standard Linux desktop machine with a 3 GHz Intel Core2-Duo processor
and 8GB RAM.
\begin{figure}[ht]
\begin{center}
\begin{tabular}{c@{\hspace{1cm}}c}
\includegraphics[width=5.5cm]{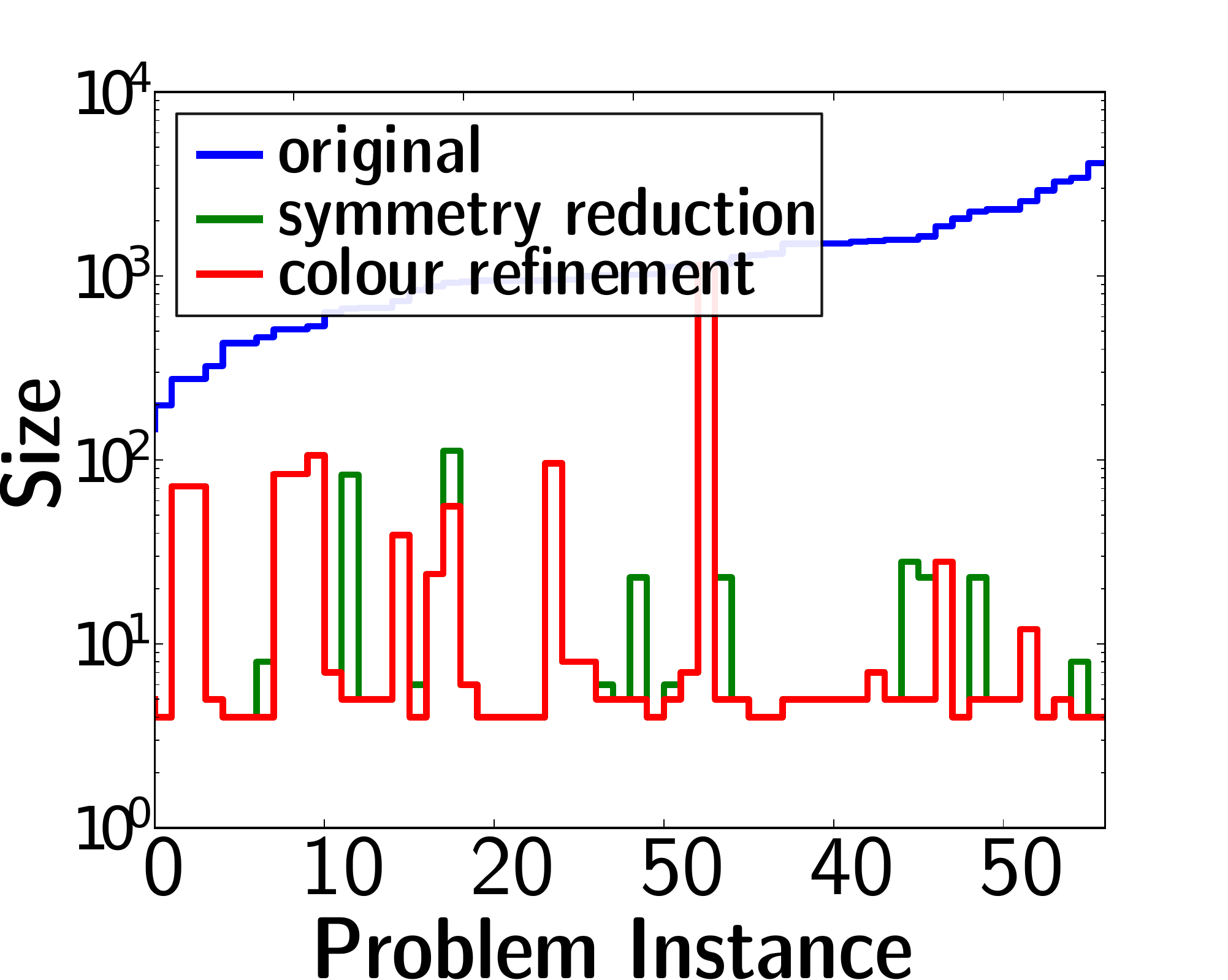}&
\includegraphics[width=5.5cm]{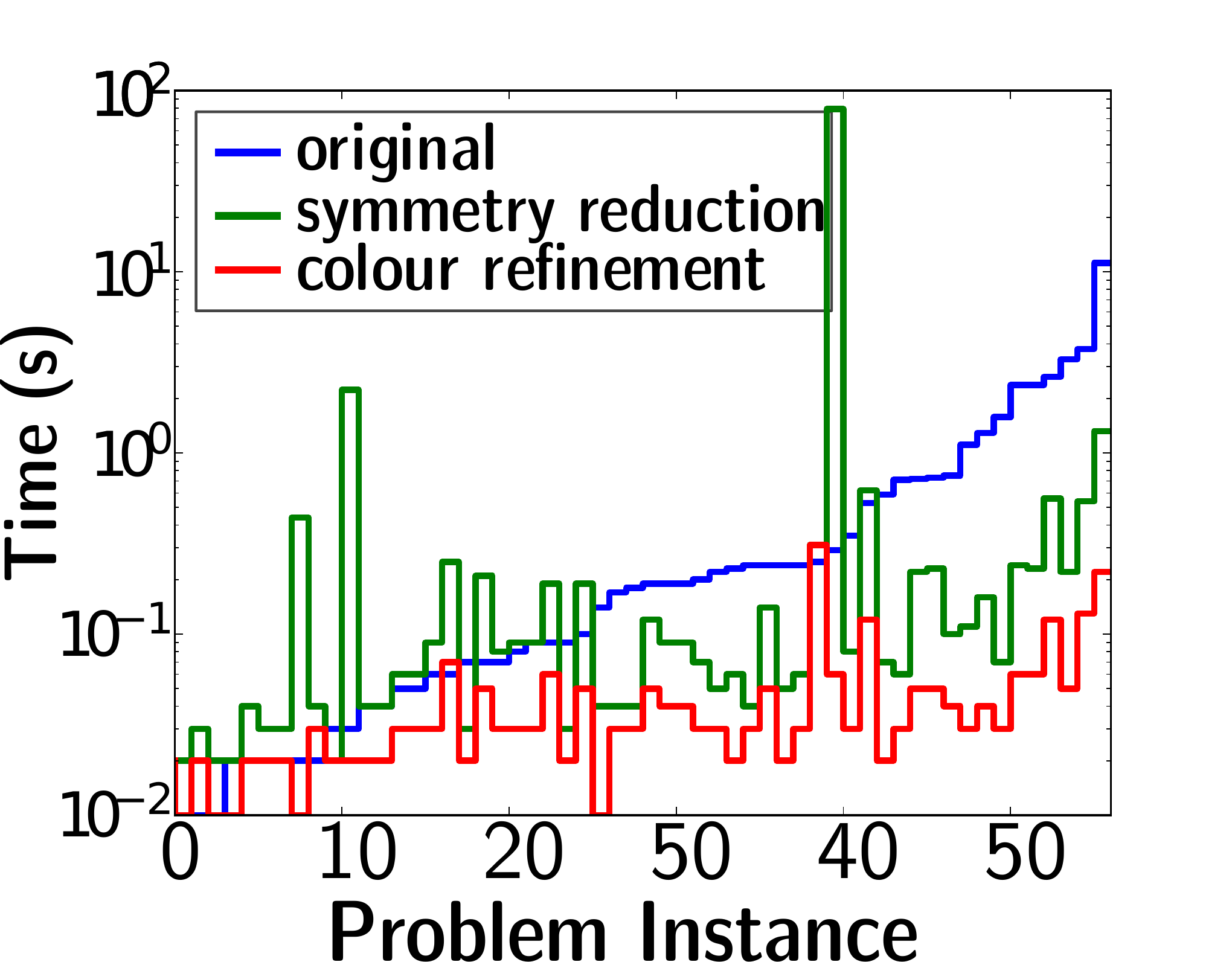}\\
(a)&(b)\\
\includegraphics[width=5.5cm]{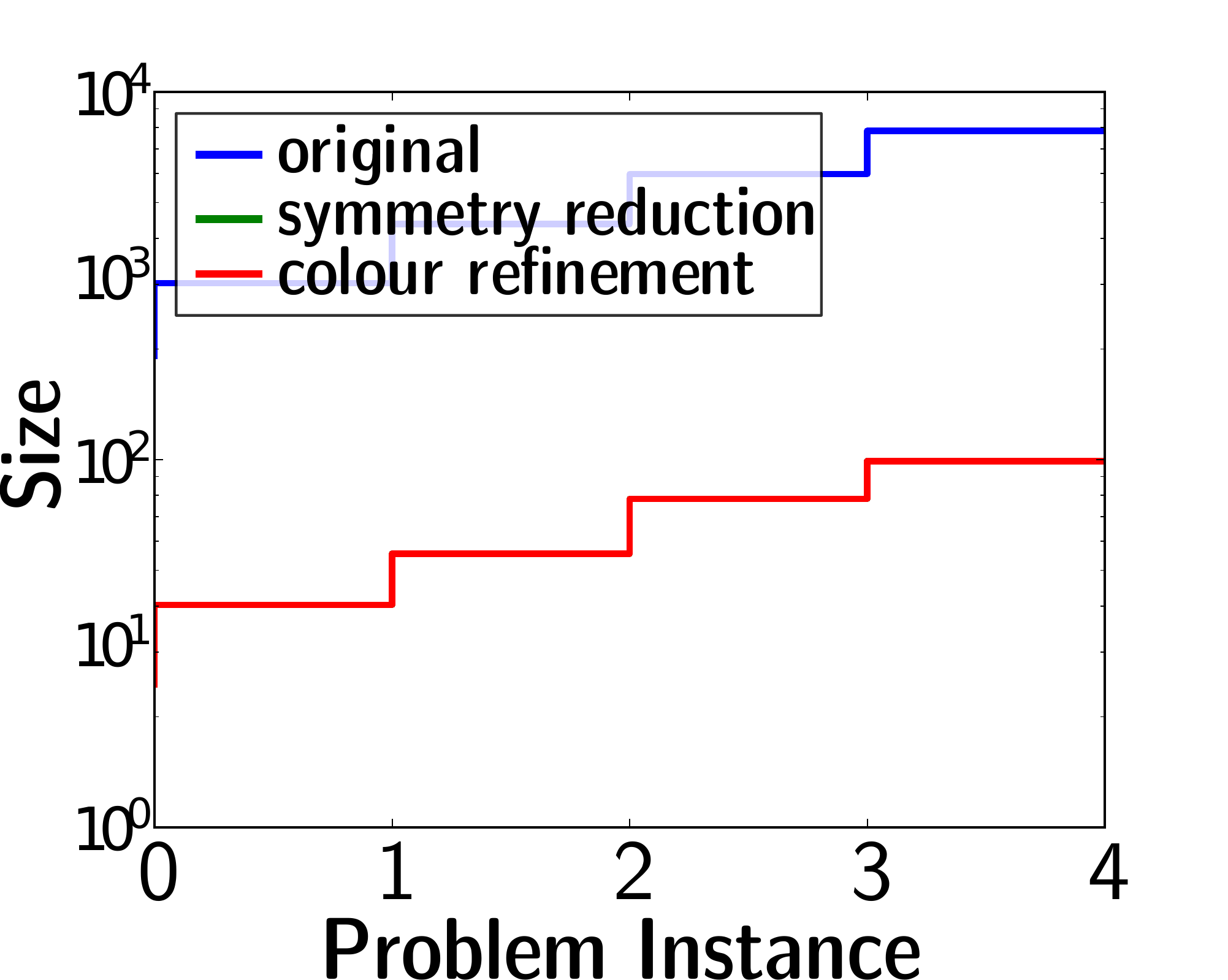}&
\includegraphics[width=5.5cm]{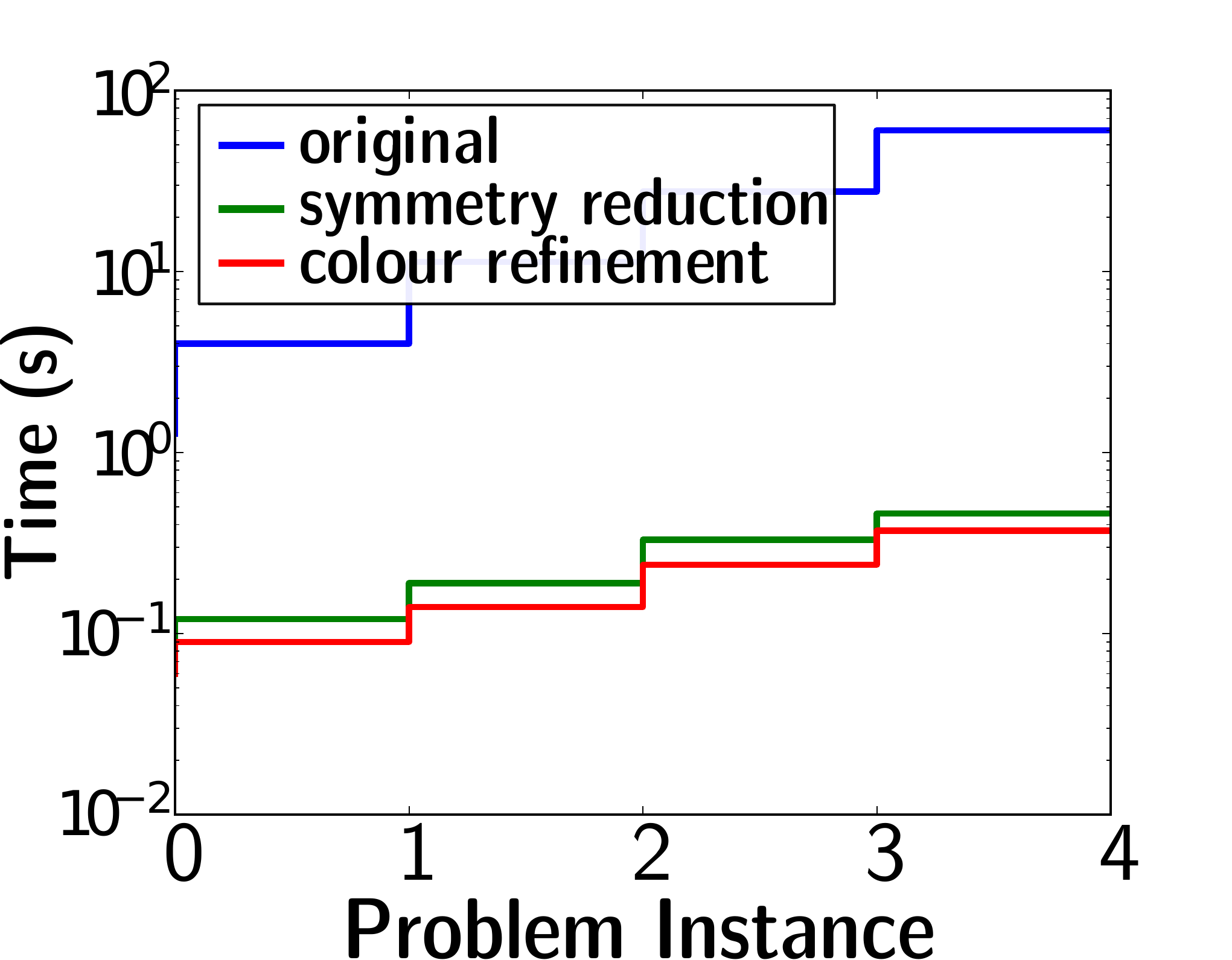}\\
(c)&(d)
\end{tabular}
\end{center}
\caption{Computational results on the different linear programs (x-axis). (a) The dimension 
(number of variables and constraints in log-scale) of the linear programs used for the 
evaluation. (b) The running times in (log-scaled) seconds (including the time for reduction) for solving the linear programs. Note that for clarity, the values in (a,b) are sorted according to the baseline independently for each figure. We refer the reader to the table in the appendix for the exact numbers.
(c,d) Same figures for computing the value functions of the grid Markov decision processes.\label{fig:exp}} 
\end{figure}

The linear programs chosen for the evaluation are relaxed versions of the integer programs 
available at Francois Margot's website \url{http://wpweb2.tepper.cmu.edu/fmargot/lpsym.html}.
They encode combinatorial optimisation problems with applications in coding theory, 
statistical design and graph theory such as computing maximum cardinality 
binary error correcting codes, edge colourings, minimum dominating sets in Hamming graphs, 
and Steiner-triple systems. 

The results are summarised in Fig.~\ref{fig:exp}(a,b). One can clearly
see that colour refinement reduces the dimension of the linear
programs at least as much as the symmetry reduction, in many cases --- as
expected --- even more. Looking at the running times, this reduction
also results in faster total computations, often an order of magnitude
faster. Overall, solving all linear programs took 38 seconds without
dimension reduction. Using the symmetry to reduce the dimensions, running
all experiments actually increased to 89 seconds, whereas using colour
refinement it only took 2 seconds. Indeed, the higher running time when
using the symmetry reduction is due to few instances only but also illustrates the
benefit of running a guaranteed quasilinear method such as colour
refinement for reducing the
dimension of linear programs.

Next, we considered 
the computation of the value function of a Markov Decision Problem  modelling decision making in situations where outcomes of actions are partly random. As 
shown in e.g.~\cite{Littman95}, the LP is $\max\nolimits_{x} \ 1^{\tr}x, \ \ \text{s.t.} \ \ x_i \leq x_i^k + \gamma\sum\nolimits_{j}p_{ij}^k x_j\;,$ where $x_i$ is the value of state $i$, $c_i^k$ is the reward that the agent receives when carrying out action $k$ in $i$, and $p_{ij}^k$ is the probability of transferring from state $i$ to state $j$ by the action $k$. The MDP instance that we used is the well-known Gridworld, see e.g.~\cite{suttonbarto98}. Here, an agent navigates within a grid of $n\times n$ states. Every state has an associated reward $R(s)$. Typically there is one or several states with high rewards, considered the goals, whereas the other states have
zero or negative associated rewards. We induced symmetries by putting a goal in every corner of the grid. 
The results for different grid sizes $n$ are summarised in Fig.~\ref{fig:exp}(c,d) and confirm our previous
results. Indeed, as expected, colour refinement and automorphisms result in the same partitions but colour refinement is 
faster. 

Finally, triggered by~\cite{mladenov12aistats}, we considered MAP inference in 
Markov logic networks (MLNs)~\cite{RichardsonDomingos:06} via 
the standard LP relaxation for MAP of the induced graphical model, 
see e.g.~\cite{globersonJ07}. Specifically, we used Richardson and Domingos' smoker-friends MLN 
encoding that friends have similar smoking habits. The so-called Frucht (among $12$ people) 
and McKay (among $8$ people) graphs were used 
to encode the social network, i.e., who are friends.   
The induced LPs were of sizes $1710$ resp. $729$. 
Solving them took $0.35$ resp.~$0.05$ seconds.
Using symmetry reduction, the sizes reduced to $1590$ resp. $247$. Reducing and solving them took $0.34$ resp $0.02$
seconds. Colour refinement, however, reduced the sizes to $46$ resp $114$. Reducing and solving the corresponding LPs took $0.02$ seconds in both cases.

%%%%%%%%%%%%%%%%%%%%%%%%%%%%%%%%%%%%%%%%%%%%%%%%%%%%%%%%%%%%
\section{Conclusions}
We develop a theory of fractional automorphisms and equitable
partitions of matrices and show how it can be used to reduce the
dimension of linear programs. 
The main point is that there is no need
to compute full symmetries (that is, automorphisms) to do a symmetry
reduction for linear programs, an equitable partition will do, and
that colour refinement can compute the coarsest equitable partition
very efficiently. We demonstrate experimentally that the
gain of our method can be significant, also in comparison with other
symmetry reduction methods.

In particular, we benefit from the fact that the colour refinement
algorithm on which we rely is very efficient, running
in quasilinear time. For
really large scale applications, however, it would be desirable to implement the
algorithm in a distributed fashion. Towards this end, in
\cite{kermlagar+14} we viewed graph isomorphism as a convex
optimisation problem and showed that colour refinement can be viewed
as a variant of the Franke-Wolfe convex optimisation algorithm. We
also gave an algorithm computing the coarsest equitable partition by a
variant of the power iteration algorithm for computing eigenvalues.

Our method works well if colour refinement has few colour classes. A
key to understanding when this happens might be Atserias and Maneva's
\cite{atsman13} notion of local linear programs. In particular, for
local linear programs we may have a substantial reduction for higher
levels of the Sherali-Adams hierarchy.

Another interesting
open question is whether there exist ``approximate versions'' of colour
refinement that can be used to solve (certain) linear programs approximately and
can be implemented even more efficiently.

% \bibliographystyle{plain}
% \bibliography{cr}

\newpage
%%%%%%%%%%%%%%%%%%%%%%%%%%%%%%%%%%%%%%%%%%%%%%%%%%%%%%%%%%%%
\appendix

%%%%%%%%%%%%%%%%%%%%%%%%%%%%%%%%%%%%%%%%%%%%%%%%%%%%%%%%%%%%
\section{Experimental Results}

The following table shows the results of our first series of
experiments with Margot's benchmark (see
\url{http://wpweb2.tepper.cmu.edu/fmargot/lpsym.html}) in some more
detail. The filenames refer to Margot's benchmark. We run three
different solvers: Columns marked ``N'' refer to the original LP
without any reduction. Columns marked ``Sr'' refer to the LP reduced
by symmetry reduction, and columns marked ``Cr'' refer to the LP
reduced by colour refinement.  We list the total time for solving the
the LPs, including the time for the reduction, the number of
variables, and the number of constraints.

\begin{small}
\begin{longtable}{l| c c c | c c c | c c c }
 & \multicolumn{3}{|c|}{Solution time} & \multicolumn{3}{|c|}{Variables} & \multicolumn{3}{|c}{Constraints} \\
Filename & N & Sr & Cr & N & Sr & Cr & N& Sr & Cr  \\
\hline
O4\_35.lp & 0.23 & 0.03 & 0.02 & 280 & 1 & 1 & 840 & 5 & 4 \\
bibd1152.lp & 0.71 & 0.21 & 0.04 & 462 & 1 & 1 & 1034 & 4 & 4 \\
bibd1154.lp & 0.72 & 0.22 & 0.04 & 462 & 1 & 1 & 1034 & 4 & 4 \\
bibd1331.lp & 0.22 & 0.05 & 0.01 & 286 & 1 & 1 & 728 & 4 & 4 \\
bibd1341.lp & 2.36 & 0.22 & 0.05 & 715 & 1 & 1 & 1586 & 4 & 4 \\
bibd1342.lp & 2.36 & 0.23 & 0.05 & 715 & 1 & 1 & 1586 & 4 & 4 \\
bibd1531.lp & 0.74 & 0.09 & 0.03 & 455 & 1 & 1 & 1120 & 4 & 4 \\
bibd738.lp & 0.0 & 0.0 & 0.01 & 35 & 1 & 1 & 112 & 4 & 4 \\
bibd933.lp & 0.01 & 0.01 & 0.0 & 84 & 1 & 1 & 240 & 4 & 4 \\
ca36243.lp & 0.01 & 0.02 & 0.01 & 64 & 1 & 1 & 368 & 3 & 3 \\
ca57245.lp & 0.05 & 0.08 & 0.02 & 128 & 1 & 1 & 816 & 3 & 3 \\
ca77247.lp & 0.06 & 0.07 & 0.02 & 128 & 1 & 1 & 816 & 3 & 3 \\
clique9.lp & 0.21 & 0.04 & 0.02 & 288 & 1 & 1 & 720 & 5 & 4 \\
cod105.lp & 11.18 & 1.31 & 0.21 & 1024 & 1 & 1 & 3072 & 3 & 3 \\
cod105r.lp & 2.62 & 0.55 & 0.11 & 638 & 3 & 3 & 1914 & 9 & 9 \\
cod83.lp & 0.19 & 0.06 & 0.02 & 256 & 1 & 1 & 768 & 3 & 3 \\
cod83r.lp & 0.16 & 0.03 & 0.02 & 219 & 6 & 6 & 657 & 18 & 18 \\
cod93.lp & 1.28 & 0.15 & 0.03 & 512 & 1 & 1 & 1536 & 3 & 3 \\
cod93r.lp & 1.1 & 0.1 & 0.02 & 466 & 7 & 7 & 1398 & 21 & 21 \\
codbt06.lp & 3.28 & 0.21 & 0.04 & 729 & 1 & 1 & 2187 & 3 & 3 \\
codbt24.lp & 0.34 & 0.07 & 0.02 & 324 & 1 & 1 & 972 & 3 & 3 \\
cov1053.lp & 0.18 & 0.11 & 0.04 & 252 & 1 & 1 & 679 & 5 & 5 \\
cov1054.lp & 0.23 & 0.13 & 0.04 & 252 & 1 & 1 & 889 & 6 & 6 \\
cov1054sb.lp & 0.24 & 0.3 & 0.3 & 252 & 252 & 252 & 898 & 898 & 898 \\
cov1075.lp & 0.05 & 0.24 & 0.06 & 120 & 1 & 1 & 877 & 7 & 7 \\
cov1076.lp & 0.06 & 0.2 & 0.04 & 120 & 1 & 1 & 835 & 7 & 7 \\
cov1174.lp & 0.52 & 0.61 & 0.11 & 330 & 1 & 1 & 1221 & 6 & 6 \\
cov954.lp & 0.04 & 0.05 & 0.02 & 126 & 1 & 1 & 507 & 6 & 6 \\
flosn52.lp & 0.17 & 0.03 & 0.02 & 234 & 4 & 1 & 780 & 19 & 4 \\
flosn60.lp & 0.23 & 0.04 & 0.01 & 270 & 4 & 1 & 900 & 19 & 4 \\
flosn84.lp & 0.58 & 0.06 & 0.01 & 378 & 4 & 1 & 1260 & 19 & 4 \\
jgt18.lp & 0.02 & 0.01 & 0.01 & 132 & 19 & 19 & 402 & 87 & 87 \\
jgt30.lp & 0.13 & 0.03 & 0.0 & 228 & 20 & 10 & 690 & 92 & 46 \\
mered.lp & 1.57 & 0.06 & 0.02 & 560 & 4 & 1 & 1680 & 19 & 4 \\
oa25332.lp & 0.18 & 0.08 & 0.03 & 243 & 1 & 1 & 1026 & 4 & 4 \\
oa25342.lp & 0.23 & 0.05 & 0.02 & 243 & 1 & 1 & 1296 & 4 & 4 \\
oa26332.lp & 3.74 & 0.53 & 0.12 & 729 & 1 & 1 & 2538 & 4 & 4 \\
oa36243.lp & 0.01 & 0.03 & 0.02 & 64 & 1 & 1 & 608 & 4 & 4 \\
oa56243.lp & 0.01 & 0.03 & 0.01 & 64 & 1 & 1 & 608 & 4 & 4 \\
oa57245.lp & 0.09 & 0.18 & 0.04 & 128 & 1 & 1 & 1376 & 4 & 4 \\
oa66234.lp & 0.01 & 0.02 & 0.01 & 64 & 16 & 16 & 212 & 56 & 56 \\
oa67233.lp & 0.03 & 0.03 & 0.01 & 128 & 20 & 20 & 384 & 64 & 64 \\
oa68233.lp & 0.18 & 0.08 & 0.03 & 256 & 24 & 24 & 698 & 72 & 72 \\
oa76234.lp & 0.0 & 0.01 & 0.0 & 64 & 16 & 16 & 212 & 56 & 56 \\
oa77233.lp & 0.03 & 0.03 & 0.01 & 128 & 20 & 20 & 384 & 64 & 64 \\
oa77247.lp & 0.08 & 0.18 & 0.05 & 128 & 1 & 1 & 1376 & 4 & 4 \\
of5\_14\_7.lp & 0.06 & 0.02 & 0.01 & 175 & 15 & 1 & 490 & 68 & 4 \\
of7\_18\_9.lp & 0.7 & 0.05 & 0.02 & 441 & 5 & 1 & 1134 & 23 & 4 \\
ofsub9.lp & 0.08 & 0.02 & 0.01 & 203 & 7 & 7 & 527 & 32 & 32 \\
pa36243.lp & 0.0 & 0.02 & 0.01 & 64 & 1 & 1 & 368 & 3 & 3 \\
pa57245.lp & 0.07 & 0.08 & 0.02 & 128 & 1 & 1 & 816 & 3 & 3 \\
pa77247.lp & 0.08 & 0.08 & 0.02 & 128 & 1 & 1 & 816 & 3 & 3 \\
sts135.lp & 0.28 & 79.18 & 0.05 & 135 & 1 & 1 & 3285 & 7 & 3 \\
sts27.lp & 0.0 & 0.01 & 0.0 & 27 & 1 & 1 & 171 & 3 & 3 \\
sts45.lp & 0.01 & 0.43 & 0.0 & 45 & 1 & 1 & 420 & 7 & 3 \\
sts63.lp & 0.02 & 2.22 & 0.01 & 63 & 1 & 1 & 777 & 5 & 3 \\
sts81.lp & 0.04 & 0.05 & 0.02 & 81 & 1 & 1 & 1242 & 3 & 3 \\
\end{longtable}
\end{small}

\end{document}